\documentclass{article}

\usepackage[thmtools-compat]{keytheorems}
\usepackage{amsmath,amssymb,mathtools}
\usepackage{hyperref}

\usepackage{zref-clever}
\newcommand{\cref}[1]{\zcref{#1}}
\newcommand{\Cref}[1]{\zcref[S]{#1}}

\usepackage{color}
\usepackage[letterpaper,margin=1in]{geometry}
\usepackage[style=alphabetic,maxnames=100,minnames=100,minalphanames=100,maxalphanames=100]{biblatex}
\usepackage{nicefrac}
\usepackage{algorithm}
\usepackage{algpseudocode}

\usepackage{makecell}
\usepackage{authblk}
\usepackage{tikz}
\usepackage{pgfplots}
\pgfplotsset{compat=1.15}
\usepackage{mathrsfs}
\usetikzlibrary{arrows}

\usepackage{comment}

\newtheorem{theorem}{Theorem}

\newtheorem{lemma}[theorem]{Lemma}
\newtheorem{definition}[theorem]{Definition}
\newtheorem{cor}[theorem]{Corollary}

\newtheorem{fact}[theorem]{Fact}

\newcommand{\R}{\mathbb R}
\newcommand{\E}{\mathbb E}
\newcommand{\N}{\mathbb N}

\newcommand{\T}{\top}

\newcommand{\sfe}{\mathbb{S}^{d-1}}
\newcommand{\dd}{\operatorname{d}}

\newcommand{\spa}{\operatorname{span}}

\newcommand{\aff}{\operatorname{affhull}}
\newcommand{\linsp}{\operatorname{span}}
\newcommand{\eps}{\varepsilon} 
\newcommand{\cone}{\operatorname{cone}}
\newcommand{\conv}{\operatorname{conv}}

\newcommand{\poly}{\operatorname{poly}}

\definecolor{sophcolor}{rgb}{0.8, 1, 0.95}

\usepackage{color}

\newcommand{\dist}{\operatorname{dist}}
\newcommand{\diam}{\operatorname{diam}}

\newcommand{\vol}{\operatorname{vol}}
\newcommand{\ball}{\mathbb{B}_2}
\DeclarePairedDelimiter\abs{\lvert}{\rvert}%
\DeclarePairedDelimiter\norm{\lVert}{\rVert}%
\newcommand{\sprod}[2]{\langle #1, #2 \rangle}

\definecolor{b2}{RGB}{51,153,255}
\definecolor{soph}{RGB}{130,17,166}
\definecolor{myorange}{RGB}{255,140,0}
\definecolor{mygreen}{RGB}{19, 173, 48}

\allowdisplaybreaks

\title{Optimal Smoothed Analysis of the Simplex Method}

\author[1]{Eleon  Bach\thanks{\href{mailto:eleon.bach@tum.de}{\texttt{eleon.bach@tum.de}}.}}
\author[2]{Sophie Huiberts\thanks{\href{mailto:sophie@huiberts.me}{\texttt{sophie@huiberts.me}}. This work was supported by ANR grant ANR-24-CE48-2762. }}

\affil[1]{Technical University of Munich}
\affil[2]{CNRS, Université Clermont Auvergne Clermont Auvergne INP, LIMOS}

\date{}

\begin{document}
\maketitle

\begin{abstract}
    Smoothed analysis is a method for analyzing the performance of
    algorithms, used especially for those algorithms whose running time in practice is significantly
    better than what can be proven through worst-case analysis.
    Spielman and Teng (STOC '01) introduced the smoothed analysis framework of algorithm analysis and applied it to the simplex method.
    Given an arbitrary linear program with $d$ variables and $n$ inequality constraints,
    Spielman and Teng proved that the simplex method runs in time $O(\sigma^{-30} d^{55} n^{86})$,
    where $\sigma > 0$ is the standard deviation of Gaussian distributed noise added to the original LP data.
    Spielman and Teng's result was simplified and strengthened over a series of works, with the current strongest upper bound being
    $O(\sigma^{-3/2} d^{13/4} \log(n)^{7/4})$ pivot steps due to Huiberts, Lee and Zhang (STOC '23).
    We prove that
    there exists a simplex method whose smoothed complexity is upper bounded by
    $O(\sigma^{-1/2} d^{11/4} \log(n)^{7/4})$ pivot steps.
    Furthermore, we prove a matching high-probability lower bound of
    $\Omega( \sigma^{-1/2} d^{1/2}\ln(4/\sigma)^{-1/4})$
    on the combinatorial diameter of the feasible polyhedron after smoothing,
    on instances using $n = \lfloor (4/\sigma)^d \rfloor$ inequality constraints.
    This lower bound indicates that our algorithm has optimal noise dependence among all simplex methods,
    up to polylogarithmic factors.
\end{abstract}

\section{Introduction}
Ever since its first use in early 1948, the simplex method has been one of the primary algorithms for solving linear programming (LP) problems.
For the purpose of this paper, an LP is a problem described by input data $A \in \R^{n \times d}, b \in \R^n, c \in \R^d$ and is written as
\begin{align*}
    \operatorname{maximize} \quad & c^\T x \\
    \operatorname{subject~to} \quad & Ax \leq b.
\end{align*}
The computational task at hand is to find if there exists any $x \in \R^d$ such that the system of inequalities $Ax \leq b$ holds.
If such a \emph{feasible solution} $x$ exists, then one must report either a feasible solution $x$ for which additionally the inner product $c^\T x$ is maximal
among all feasible solutions, or a certificate that the set of feasible solutions is \emph{unbounded}.
The set of feasible solutions $P \coloneqq \{x \in \R^d \colon Ax \leq b \}$ is called a polyhedron.

Linear programming problems arise in innumerable industrial contexts. Furthermore, they are often used as fundamental steps in a vast range of other optimization algorithms as they are well-known to be solvable efficiently.
Despite the tremendous progress for polynomial time methods in general \cite{khachiyan1980polynomial} and interior point methods in particular \cite{ karmarkar1984new, renegar1988polynomial, mehrotra1992implementation, lee2014path, adlnv}, the simplex method remains one of the most popular algorithms to solve LPs in a wide variety of practical contexts.

The simplex method is best thought of as a class of algorithms, differing in specific details such as the choice of the \emph{pivot rule}
or the \emph{phase 1} procedure.
Navigating from one vertex of the feasible set to another, the pivot rule is the part of a simplex method that decides
in which direction the pivot step will move.
Notable examples of pivot rules include the most negative reduced cost rule \cite{dan51}, the steepest edge rule and its approximations \cite{harris, gol76, jour/mapr/FG92},
and the shadow vertex rule \cite{gas55, thesis/Borgwardt77}.

Although there have been substantial improvements over the simplex method as it was first introduced by Dantzig, one thing has not changed:
the total number of pivot steps required to solve LPs in practice scales roughly linear in the dimensions of the problem \cite{dantzigbook,sha87, andrei2004complexity, xpressmanual}.
Despite many decades of practical experience supporting this observation,
it remains a major challenge for the theory of algorithms to explain this phenomenon.
This is further complicated by the results from worst-case analysis: for almost every major pivot rule,
there are theoretical constructions known that make the simplex method take exponentially many pivot steps before reaching an optimal solution.
Because of the misleading results of worst-case analysis, the simplex method has been a showcase for the development of new methods to go beyond worst case analysis.

The majority of these worst-case constructions are based on \emph{deformed products}
\cite{km72, jer73, AC78, GS79, Mur80, g83, jour/cm/AZ98} and extending arguments \cite{black2024exponentiallowerboundspivot}
or
on Markov decision processes
\cite{conf/stoc/FHZ11, conf/ipco/Friedmann11, disser2020exponential, dissermosis}.
The fastest known simplex method under the worst-case analysis paradigm is randomized and requires $2^{O(\sqrt{d \log(1+n/d)})}$ pivot steps \cite{k92,msw96,hz15} in expectation.
The simplex method was shown to have polynomial running time for classes of polytopes such as $0/1$-polytopes \cite{black2024simplex,black22},
bounded subdeterminants  \cite{conf/icalp/BR13, dadushhahnle},
and bounded ratios of non-zero slack values \cite{kitaharamizuno2011}.

During the 70's and 80's, there were a number of investigations into the average-case complexity
of the simplex method.
A wide variety of models was studied, including drawing the rows of $A$ from
a spherically-symmetric distribution \cite{thesis/Borgwardt77, b82,b87,b99,bdghl21},
drawing the combined vector $(c,b)$ from a spherically symmetric distribution \cite{jour/mapr/Smale83},
having fixed $A,b$ and every inequality constraint independently being either $a_i^\T x \leq b_i$
or $a_i^\T \geq b_i$ \cite{report/Haimovich83},
and a range of other models \cite{jour/jacm/AM85, jour/mapr/Megiddo86, jour/mapr/Todd86, jour/jc/AKS87}.
For an in-depth survey we refer the reader to \cite{b87}.

A major weakness of average-case analysis is that real-life LPs are structured in recognizable ways,
whereas average-case LPs have no such structure.
As such, it is reasonable to question to what extent average-case analyses succeed at explaining
the simplex method's performance in practice.
Smoothed analysis is a more sophisticated way of going beyond worst-case analysis \cite{bwcabook},
drawing on the advantages of average-case analysis while still preserving
the large-scale geometric structure in the input instances.
It is commonly understood to demonstrate that inputs on which the simplex
method performs badly are ``pathological'', in the sense that they depend on very
brittle small-scale structures in their input data.
The focus of recent works in smoothed analysis is on improving the dependence on the noise parameter,
which can then be interpreted as measuring ``how brittle'' these structures are.
Beyond the simplex method, smoothed analysis has been applied to a wide range of popular algorithms.
Examples of this include
multicriteria optimization \cite{BV03,v010a010,Beier2022},
Lloyd's $k$-mean algorithm \cite{AMR11},
the 2-OPT heuristic for the TSP \cite{ERV14, KMV23, MvR23},
local Max-Cut \cite{localmaxcut,Chen2024},
makespan scheduling \cite{Rohwedder2025},
policy iteration for MDPs \cite{miranda}
and many more.

One assumes that a base LP problem is adversarially constructed
\begin{align*}
    \operatorname{maximize} \quad& c^\T x \\
    \operatorname{subject~to} \quad& \bar A x \leq \bar b,
\end{align*}
with the assumption that $\bar A \in \R^{n \times d}$ and $\bar b \in \R^n$
are such that the rows of the combined matrix $(\bar A, \bar b)$ each
have Euclidean norm at most $1$.
Subsequently, this input data gets randomly perturbed. For a parameter $\sigma > 0$,
one samples $\hat A \in \R^{n \times d}$ and $\hat b \in \R^n$ with independent entries,
each entry being drawn from a Gaussian distribution with mean $0$ and variance $\sigma^2$.
The \emph{smoothed complexity} of an algorithm is the expected running time to solve the perturbed problem
\begin{align*} \tag{Input LP}\label{input-LP}
    \operatorname{maximize} \quad& c^\T x \\
    \operatorname{subject~to} \quad& (\bar A + \hat A) x \leq \bar b + \hat b,
\end{align*}
where the running time is to be bounded as a polynomial function in $n,d$ and $\sigma^{-1}$.
The dependence on $\sigma$ is key to smoothed analysis.
When $\sigma$ is large enough such that $\hat A$ dominates $\bar A$, the smoothed complexity converges to the average case as analyzed by Borgwardt \cite{thesis/Borgwardt77}.
If, on the contrary, $\sigma$ goes to $0$, we find ourselves in the situation to analyze the worst-case complexity of the simplex method.
By choosing small $\sigma$, for example inversely polynomial in $n$, smoothed analysis combines advantages of both worst-case and average-case analysis.
A low smoothed complexity for an algorithm is thought to mean that one should expect this algorithm to perform well in practice.
One proposed reason is that many real-world instances are generated from data that is inherently noisy by nature.
In further contrast to average-case analysis, we observe that the probability mass used in \eqref{input-LP} is concentrated in a region of radius $O(\sigma \sqrt{d \ln(n/d)})$.
When $\sigma$ is small, this region contains an exponentially small fraction of the probability mass as considered in the average-case analysis by Borgwardt \cite{thesis/Borgwardt77}.
Smoothed analysis results are thus much stronger than comparable average-case analysis results.
The case when $\sigma$ is large is considered less interesting due to it not preserving the original structure in the instance.
To illustrate this point, consider the following: every row of $\bar A$ is assumed to have Euclidean norm at most $1$.
If $\sigma > 1/\sqrt{d}$ then every row of $A = \bar A + \hat A$ is expected to consist of more noise than structure.
Hence the case of small $\sigma$ is most commonly studied, with many papers reporting in their final
running time bound only the term with highest dependence on $\sigma$.

Spielman and Teng's algorithm is based on the \emph{shadow vertex pivot rule}. This pivot rule works by
having two objectives $c,c'\in \R^d$ and visiting all basic solutions that
maximize some positive linear combination of the two.
Starting from an optimal basic feasible solution to the first objective,
it pivots until it finds an optimal basic feasible solution to the second objective
(or finds an infinite ray certifying unboundednes).
When the two objectives are chosen independently of the noise $\hat A$, and the right-hand side vector
is the unperturbed all-ones vector,
then the number of pivot steps required by the shadow vertex rule
is naturally upper bounded by the \emph{shadow size} $D(n,d,\sigma)$. The shadow size $D(n,d,\sigma)$ is the expected number of vertices of the orthogonal projection
$\pi_W(\{x : (\bar A + \hat A)x \leq \mathbf{1}\})$ onto $W \eqqcolon \spa(c,c')$.
Spielman and Teng's result is made up of two parts.
First an analysis of the shadow size $ D(n,d,\sigma)$.
Formally, they bound the quantity
\[
    D(n,d,\sigma) ~=~
        \max_{\bar A, c, c'}~
        \E_{\hat A} \left[~
            \operatorname{vertices}\left(\pi_{\linsp(c,c')}(\{x : (\bar A + \hat A)x \leq \mathbf{1}\})\right)~
            \right],
\]
where $\bar A$ is assumed to have rows each of Euclidean norm at most $1$.
Hence, they bound the number of pivot steps taken by the shadow vertex simplex method
when moving on the set $\{x : (\bar A + \hat A)x \leq \mathbf{1}\}$ from the maximizer of
a fixed objective $c$ to the maximizer of another fixed objective $c'$.

The second part is an algorithmic reduction, showing that there exists a simplex method
whose running time can be bounded as a function of $ D(n,d,\sigma)$.
They proved a bound on the shadow size of
\[
     D(n,d,\sigma) \leq \frac{10^8 n d^3}{\min(\sigma, 1/3\sqrt{d \ln n})^6},
\]
and they found a simplex method that requires an estimated
\[
    O\left(nd\ln(n/\min(1,\sigma)\right)  D\left(n,d,\frac{\min(1,\sigma^5)}{d^{8.5}n^{14}\ln^{2.5}n}\right)
\]
pivot steps under the smoothed analysis framework.
This combines to a total of $\hat O(n^{86}d^{55}\sigma^{-30})$ pivot steps, ignoring logarithmic factors and
assuming that $\sigma \leq 1/3\sqrt{d\ln n}$. Note that this last assumption on $\sigma$ may be made without loss of generality since we can scale down the constraints of the LP to make the assumption hold true.
The result of this scaling can be captured in an additive term in the upper bound that is independent of $\sigma$.

Spielman and Teng's work was built upon by \cite{ds05}, who improved the shadow bound to
\[
    D(n,d,\sigma) \leq \frac{10^4 n^2 d \ln n}{\sigma^2} + 10^5 n^2 d^2 \ln^2 n.
\]
Vershynin later proved in \cite{ver09} a shadow bound of $D(n,d,\sigma) \leq d^3 \sigma^{-4} + d^5 \ln^2 n$, dramatically improving the dependence on $n$ to poly-logarithmic, which is the state-of-the-art dependence of $n$ up to today.
The price he paid, however, was a worse dependence on the noise parameter $\sigma$.
Finding the optimal dependence on $\sigma$ while maintaining the poly-logarithmic dependence on $n$, has been the objective of smoothed analysis of the shadow vertex method follow-up work ever since.
Vershynin found an alternative algorithm running in time $O(D(n+d,d,\min(\sigma, 1/\sqrt{d \ln n}, 1/d^{3/2}\ln d)))$.
With the two works \cite{ds05} and \cite{ver09}, there was a situation where one bound on $D(n,d,\sigma)$ had better dependence on
$\sigma$ and the other had much better dependence on $n$.
This was resolved with the work of \cite{DH18},
who proved a best-of-both-worlds bound of $D(n,d,\sigma) \leq d^2\sigma^{-2}\sqrt{\ln n} + d^{3}\ln^{1.5} n$.
They also observed that the comparatively simple \emph{dimension-by-dimension}
phase 1 algorithm of \cite{b87} could be used with an expected number of pivot
steps of at most $(d+1)D(n,d+1,\sigma)$.
The same observation had been previously and independently made in \cite{thesis/Schnalzger14}.

The shadow bound with the current best dependence on $\sigma$ comes from \cite{hlz} and states that
\begin{align} \label{al:hlz_result}
    D(n,d,\sigma) \leq O\left(\frac{d^{13/4}\ln^{7/4} n}{\sigma^{3/2}} + d^{19/4}\ln^{5/2} n\right).
\end{align}
The same paper also proved the first non-trivial lower bound, stating that
\[
    D(4d-13,d,\sigma) \geq \Omega\left(\min\left(2^d, \frac{1}{\sqrt{d\sigma\sqrt{\log d}}}\right)\right)
\]
by constructing explicit data $\bar A, c, c'$, assuming that $d \geq 5$.
Based on a computational experiment, they conjecture that their construction
might have smoothed shadow sizes as large as $\sigma^{-3/4}/\poly(d)$.
In all these works, the complexity of the algorithms is reduced in a black box manner to shadow bounds of smoothed unit LP, i.e., LPs of the form $(\bar A + \hat A) x \leq \mathbf{1}$.

\subsection{Our results} Our main contribution is a substantially improved shadow bound presented below. An overview about our and previous upper and lower bounds are summarized in \Cref{tab:shadow_history}.
\begin{table}[H]
  \begin{center}
  \renewcommand{\arraystretch}{1.6}
    \begin{tabular}{|l|l|l|l|} 
        \hline
        {} & {\bf Expected number of vertices} & {\bf Model} \\
        \hline
        Borgwardt & $\Theta(d^{3/2}\sqrt{\log n})$ & \makecell[l]{Average-\\ case,\\ Gaussian } \\
        \hline
        \hline
        Spielman, Teng'04 & $O(\sigma^{-6}d^3n + d^6n\log^3 n)$ & $D(n,d,\sigma)$\\
        \hline
        Deshpande, Spielman'05 & $O(\sigma^{-2}dn^2\log n + d^2n^2\log^2 n)$ & $D(n,d,\sigma)$\\
        \hline
        Vershynin'09 & $O(\sigma^{-4}d^3 + d^5\log^2n)$ & $D(n,d,\sigma)$\\
        \hline
       Dadush, Huiberts'18 & $O(\sigma^{-2}d^2\sqrt{\log n}+ d^{3}\log ^{1.5}n)$ & $D(n,d,\sigma)$\\
        \hline
        Huiberts, Lee, Zhang'23 & $O(\sigma^{-3/2} d^{13/4} \log^{7/4} n + d^{19/4} \log^{13/4} n)$ & $D(n,d,\sigma)$\\
        \hline
        This paper & $O(\sigma^{-1/2} d^{11/4} \log(n)^{7/4} + d^3 \log(n)^2)$ & $R(n,d,\sigma)$\\
        \hline
        \hline
        Huiberts, Lee, Zhang'23 & $\Omega(\min(\frac{1}{\sqrt{\sigma d \sqrt{\log d}}},2^d))$ & $D(4d-13,d,\sigma)$\\
        \hline
        This paper & $\Omega(\frac{\sqrt{d}}{\sqrt{\sigma\sqrt{\ln(4/\sigma)}}})$ & $R(\lfloor (4/\sigma)^d \rfloor,d,\sigma)$\\
        \hline
    \end{tabular}
    \caption{Bounds of expected number of pivots in previous literature, assuming $d \geq 3$. Logarithmic factors are simplified. The lower bound of \cite{b87} holds in the smoothed models as well.}\label{tab:shadow_history}
  \end{center}
\end{table}

We provide a novel three-phase shadow-vertex simplex algorithm that relies on a new quantity that we call
the \emph{semi-random shadow size} $R(n,d,\sigma)$.
We improve the algorithmic reduction, obtaining an algorithm whose running time is
$O(R(n+d,d,\min \{\sigma,1/\sqrt{d\ln n},1/d^{3/2}\log d\}))$, where $R(n,d,\sigma)$ is defined as
\[
    R(n,d,\sigma) ~=~
        \max_{\bar A, \bar b, c}~
        \E_{\hat A, \hat b, Z} \left[~
            \operatorname{vertices}\left(\pi_{\linsp(c,Z)}(\{x : (\bar A + \hat A)x \leq \bar b + \hat b\})\right)~
            \right].
\]
Here, $\bar A, \bar b$ are again chosen such that the rows of $(\bar A,\bar b)$ each have norm at most $1$,
$c \in \sfe$ is a unit vector, $\hat A, \hat b$ have independent entries that are Gaussian distributed
with mean $0$ and standard deviation $\sigma$, and non-zero $Z \in \R^d$ is independently sampled
from any spherically symmetric distribution.
The semi-random shadow size is used to bound the number of pivot steps taken by the shadow vertex simplex method
when moving on the set $\{x : (\bar A + \hat A)x \leq \bar b + \hat b\}$ from the maximizer of
a fixed objective $c$ to the maximizer of a randomly sampled objective $Z$ (or the other way around).

Having our algorithm be able to sample $Z$ at random is the key algorithmic improvement which allows us to
prove stronger bounds than was possible using $D(n,d,\sigma)$.
Specifically, we find
\[
    R(n,d,\sigma) \leq O\left(\sqrt{\sigma^{-1} \sqrt{d^{11} \log(n)^7}} + d^3 \log(n)^2 \right).
\]
Notably, this upper bound is lower than the conjectured $\sigma^{-3/4}$ lower bound of \cite{hlz} for the fixed-plane shadow size $D(n,d,\sigma)$.
In terms of the exponent on $\sigma$ this is the best that a shadow size bound can be,
as we demonstrate in \Cref{sec:lb} with a nearly-matching lower bound of
\[
    \frac{\sqrt{d-1} }{96\sqrt{\sigma\sqrt{\ln(4/\sigma)}}} \leq R\left(\lfloor (4/\sigma)^d\rfloor,d,\sigma\right).
\]

In our proofs, we overcome a number of key challenges which were featured prominently in previous smoothed analyses of the simplex method.
The main challenge is that the (previously deterministically chosen) objectives might give rise to a (nearly) degenerate projection.
Previously, this fate was prevented only by the randomness in the Gaussian noise added to the constraint data.
The expected degree of degeneracy was quantified by the following two steps.
In the first step, \cite{hlz} proves that a typical basis visited by the algorithm is relatively well-conditioned.
In the second step, they prove that when a well-conditioned basis is visited, the shadow plane is probably not too close to any 2-dimensional face of its normal cone.
Both of these steps require the use of the noise in the constraint data, which entails that both steps lose a factor of at least $\sigma^{-1/2}$.
Our analysis avoids this two-step procedure by incorporating the necessary randomness in the pivot rule directly.
Our simplex method is based on a semi-random shadow plane which is the linear span of the objective $c$ and a randomly sampled objective $Z$.
Compared to the analysis of \cite{hlz}, this improves our bound by a factor $\sigma^{-1}$.

More detail on the techniques used to prove the upper bound can be found in \Cref{sec:overview}.

\subsection{Related Work}

The semi-random shadow vertex method used in this paper is a simplified version of that of \cite{dadushhahnle}.
They give a simplex method which uses an expected number of $O(\frac{d^2}{\delta}\ln(d/\delta))$ pivot steps,
where $\delta$ is a parameter of the constraint matrix which measures the curvature of the feasible region.
Their work improves over a weaker result that used ``less random'' shadows \cite{conf/icalp/BR13}.

The shadow size for polyhedra all whose vertices are integral was studied in
\cite{black2024simplex,black22,black2024randomshadowsfixedpolytopes},
the last of which studies uniformly random planes.
Semi-random shadow planes were also used in a different context to obtain a weakly polynomial-time
``simplex-like'' algorithm for LP in \cite{conf/stoc/KS06}.

Shadow bounds for random objectives were previously used by \cite{narayanan2021spectral}
in order to derive results similar to diameter bounds for smoothed polyhedra.
Specifically they proved that with high probability there exists a large subset of vertices (according to some specific measure)
which has small diameter.
Here the random objectives were sampled from some non-uniform distribution,
and the sizes of the resulting shadows were bounded using the shadow bound of \cite{DH18}.

In this paper we make use of a notion of vertices being ``well-separated'' from each other.
That assumption was pioneered in a line of work starting with \cite{kitaharamizuno2011}.
Different from work in the smoothed analysis paradigm, in \cite{kitaharamizuno2011}
the data is deterministic, the pivot rule is that of the most negative reduced cost,
and progress is measured with respect to the objective value.
The well-separatedness of \cite{kitaharamizuno2011} is locally similar to the polar concept of
\emph{vertex-facet distance} of \cite{hlz}. By simultaneously generalizing both, we are able to use proof techniques
borrowed from both lines of work as part of our proofs.
These proofs were developed in parallel for both the present work and \cite{bbhk}.

\subsection{Proof Overview}
\label{sec:overview}

\subsubsection{Upper bound}

The main hurdle of a smoothed analysis of the shadow vertex simplex method is to find techniques allowing for an efficient translation of progress measure on the polyhedron $P$ induced by the LP into progress measure on its $2$-dimensional shadow polygon $Q$.
In the following we will outline our new strategy for constructing a progress measure which we will refer to as ``separation'' and illustrate how this form of separation enables us to derive our new upper bounds.

Our main algorithmic improvement is that we avoid that the shadow plane is likely to be aligned with the smoothed polyhedron $P$. Handling this issue was a main challenge in smoothed analysis of the simplex method ever since Spielman and Teng's first smoothed analysis of the simplex method \cite{ST04}.
We propose an analysis that is based on a semi-random shadow plane which we will explain first.

Whereas previous \cite{ST04, ds05, ver09, DH18, hlz} smoothed analyses of the shadow vertex simplex method analyzed the geometry of the polar polygon $Q \cap W$, where $W$ denotes the shadow plane, we remain in primal space and analyze the geometry of the primal shadow polygon.

\paragraph{Random objective and semi-random shadow plane}
For two vectors $c,c' \in \R^d$, we write $\pi_{c,c'} : \R^d \to \linsp(c,c')$ for the orthogonal projection onto $\linsp(c,c')$.
Since the length of the shadow path from $Z$ to $c$ depends only on the direction $Z/\norm{Z}$ and not on the norm $\norm{Z}$
(see \Cref{lem:shadowbasics}),
for the purposes of analysis we assume that $Z \in \R^d$ is exponentially distributed,
i.e., such that for any measurable $S \subseteq \R^d$ it holds that
$\Pr[Z \in S] = \frac{1}{d! \vol_d(\ball^d)}\int_S e^{-\norm{z}}\dd z$.

For most of our analysis we need a certain amount of randomness in the objectives which we traverse the shadow path with.
Once we have a slightly-random objective $c + Z/2^k$ on the shadow path that is close enough to our fixed objective $c$,
we can show that there is only a constant number of pivot steps from $c  + Z/2^k$ to $c$ (\Cref{cor:small-angle-pivot-count}). We parameterize this closeness by some $k \in \N$.
We observe that the path from $Z$ to $c + Z/2^k$ is equal to the path from $Z$ to $2^k c + Z$.
We construct $k$ intermediate objectives $Z+2^{i}c$ for $i = 0, \dots, k$ and traverse step by step the shadow path from $Z+2^{i-1}c$ to $Z+2^{i}c$ in order to get strong control over the lengths of the shadow subpaths.
Breaking the path up into pieces will help to homogenize an important part of the argument, as discussed in the next paragraph.

In \Cref{lem:replacement-angle-bound} we will show, using an argument similar to the angle bound of \cite{ST04}, that it suffices to traverse the shadow path until $k$ is of order $O(d\log(n))$.
Note that the angle between two such vectors $Z+2^{i-1}c$ and $Z + 2^{i}c$ decreases exponentially as a function of $i$.

\paragraph{Separation}

Our progress measure consist of two components.
When traversing the shadow path from $Z$ to $2^k c + Z$, we will see in
\Cref{sec:multipliers} that for $99\%$ of the traversed bases $I$ there exists an intermediate objective
$y^I \in [Z,2^k c + Z]$ such that $(y^I)^\T A_I^{-1} \geq 0.005/d \eqqcolon m \mathbf{1}$ using a result of Bach, Black, Huiberts and Kafer \cite{bbhk}. We call this the ``good multiplier" property.
This fact is independent of the noise on the constraint data
and is proven using only the randomness in our random objective $Z$. The good multipliers are the ingredient that allows us to easily incorporate the randomness of our algorithm into the analysis.

At the same time, using the randomness in the perturbations on the constraint data, for at least $80\%$
of bases $I$ on the path, the feasible solution $x^I = A_I^{-1} b_I$ has slack at least
$b_j - a_j^\T x^I \geq \frac{\sigma \norm{x^I}}{5000 d^{3/2} \ln(n)^{3/2}} \eqqcolon g \norm{x^I}$ for every nonbasic constraint
$j \in [n]\setminus I$.
This fact is established in \Cref{sec:slack}.

The majority of traversed bases satisfies both the ``good multiplier'' and ``good slack'' criteria, from which we deduce that for these we have ``good'' vertex-neighbor separation in the following sense.
For the purpose of this sketch we assume that both properties hold for all bases on the shadow path.
The machinery that allows us to assume this is described in \Cref{sec:triples}.

Any of the above mentioned intermediate objective
$y^I \in [Z,2^k c + Z]$ certifies large distance between $x^I$ and its closest neighbor $x^J$: an easy computation shows us that
$(y^I/ \norm{y^I})^\T (x^I - x^{J}) = (y^I / \norm{y^I})^\T A_{I}^{-1} A_{I} (x^I - x^{J})   \geq  (m \cdot g) \norm{x^I} / \norm{y^I} $.
For the sake of simplicity in this proof sketch, we will consider only the case that $\norm{y^I} \leq O(d)$.
In the full proof, $\norm{y^I}$ can be exponentially large. This requires that we homogenize with respect to $\norm{y^I}$ when the algorithm traverses further along the shadow path from $Z$ to $2^k c + Z$.
This type of homogenization is new to smoothed analysis for the simplex method, but our use is inspired by \cite{dadushhahnle}.

Following the strategy proposed by Huiberts, Lee and Zhang \cite{hlz}, we translate this property via a case distinction into either certifying large edge lengths or large exterior angles, as we will explain in the following.
\begin{figure}
    \centering
    \definecolor{qqwuqq}{rgb}{0,0.39215686274509803,0}
    \definecolor{uuuuuu}{rgb}{0.26666666666666666,0.26666666666666666,0.26666666666666666}
    \definecolor{ududff}{rgb}{0.30196078431372547,0.30196078431372547,1}
    \definecolor{xdxdff}{rgb}{0.49019607843137253,0.49019607843137253,1}
    \begin{tikzpicture}[line cap=round,line join=round,>=triangle 45,x=0.8cm,y=0.8cm]
    \clip(-0.620007729313798,-2.0911247374460724) rectangle (16.632395676757376,5.590331897182605);
    \draw [shift={(6,4)},line width=2pt,color=qqwuqq,fill=qqwuqq,fill opacity=0.10000000149011612] (0,0) -- (63.66980504598516:1) arc (63.66980504598516:123.6900675259798:1) -- cycle;
    \draw [shift={(0,0)},line width=2pt,color=qqwuqq,fill=qqwuqq,fill opacity=0.10000000149011612] (0,0) -- (0:1) arc (0:33.690067525979785:1) -- cycle;
    \draw [line width=2pt,dash pattern=on 5pt off 5pt,domain=-0.620007729313798:16.632395676757376] plot(\x,{(-0-0*\x)/20});
    \draw [line width=2pt] (0,0)-- (6,4);
    \draw [line width=2pt] (6,4)-- (15.973200724483275,-0.9356042273284247);
    \draw [line width=2pt,dotted] (6,36.40001951157273)-- (6,0);
        \draw[color=xdxdff] (0.3,-0.35) node {$p_1$};
        \draw[color=qqwuqq] (5.7,5.3) node {$\alpha_2$};
        \draw[color=ududff] (6.8,4.1) node {$p_2$};
        \draw[color=ududff] (6,-0.4) node {$q$};
        \draw[color=ududff] (15.5,-1.2) node {$p_3$};
        \draw[color=black] (2.5,2.4) node {$\leq \rho \norm{p_2}$};
        \begin{small}
            \draw[color=black] (6.8,1.2) node {$\geq \eps \norm{p_2}$};
        \end{small}
    \begin{scriptsize}

    \draw [fill=xdxdff] (0,0) circle (2.5pt);
    \draw [fill=ududff] (6,4) circle (2.5pt);
    \draw [fill=ududff] (15.973200724483275,-0.9356042273284247) circle (2.5pt);
    \draw [fill=xdxdff] (6,0) circle (2pt);
    \end{scriptsize}
    \end{tikzpicture}
    \caption{Lower bounding the exterior angle at $p_2$. The intermediate objective $y$ points straight upwards. The angles at $q$ are right angles.}
    \label{fig:shadowsimplified}
\end{figure}

Consider two consecutive vertices on the shadow polygon $Q = \pi_{c,Z}(P)$.
We call these vertices $p_1 = \pi_{c,Z}(x^I)$ and $p_2 = \pi_{c,Z}(x^J)$.
Let $p_3$ denote the next consecutive vertex of $Q$, and let $y = y_J$ be an intermediate objective as described above.
For a number $\rho > 0$ to be decided later, we distinguish the cases whether $\norm{p_1 - p_2} > \rho \norm{p_2}$ or $\norm{p_1 - p_2} \leq \rho \norm{p_2}$.

In the former case, the edge $[p_1,p_2]$ ``takes up a lot of perimeter'',
in the sense that out of the integral $\int_{\partial Q} \norm{t}^{-1} \dd t$,
at least $\Omega(\rho)$ of its value is contributed by the line segment
$\int_{[p_1, p_2]} \norm{t}^{-1} \dd t \geq \Omega(\rho)$.
Assuming $\sigma > n^{-10d}$, we upper bound the full integral by $\int_{\partial Q} \norm{t}^{-1} \dd t = O(d \log n)$,
which then effectively bounds from above the number of
edges for which $\norm{p_1 - p_2} \geq \rho \norm{p_2}$ can hold by $O(d \rho^{-1} \log n)$.

In the latter case, consider the triangle with vertices $p_1, p_2$ and $q = p_2 - y^\T (p_2 - p_1) \cdot y$
as depicted in \Cref{fig:shadowsimplified}. Because the next vertex $p_3$ on the
boundary after $[p_1,p_2]$ satisfies $y^\T p_3 < y^\T p_2$, the exterior angle $\alpha_2$ at
$p_2$ is at least as large as $\angle(p_2,p_1,q)$.
The right-angled triangle has a hypotenuse of length $\norm{p_1-p_2} \leq \rho \norm{p_2}$
and an opposite side of length $\norm{p_2 - q} \geq  m \cdot g \norm{p_2} / \norm{y}  \eqqcolon \eps \norm{p_2}$,
from which we derive a lower bound on the exterior angle at $p_2$ of
\[
    \alpha_2 \geq \angle(p_2,p_1,q) \geq \sin(\angle(p_2,p_1,q)) \geq \eps/\rho.
\]
Since the sum of the exterior angles of all the vertices of $Q$ is equal to $2\pi$,
that means that there can be at most $2\pi\rho/\eps$ vertices with exterior angle that large.

This argument shows that, under our separation assumption, every vertex must either have a long edge or a large exterior angle,
hence the polygon $Q$ can only have at most $\min_{\rho > 0} 2\pi\rho/\eps + O(d \rho^{-1} \log n)$ vertices.
Choosing $\rho=\sqrt{\eps d \log n}$ yields an upper bound of $O(\sqrt{d \log(n)/\eps})$ vertices.

The above argument hides the issue of how we deal with the case if $\norm{y^I}$ is very large as $y^I \in [2^{i-1} c + Z,2^i c + Z]$, $i \in [k]$ can be very large.
The key impact of the semi-randomness of the shadow plane is that for $y^I \in [2^{i-1} c + Z, 2^i c + Z]$ of large norm, $i \in [k]$,
the largest exterior angle that can be argued is of order
\[
    \theta_I \geq \frac{\sigma}{\poly(d,\log n) \cdot \rho \cdot \norm{y}} \approx \frac{\sigma}{\poly{(d,\log n) \cdot \rho \cdot \norm{2^i c + Z}}}.
\]
Hence, the exterior angles shrink as the objectives grow.
However, the total angle to be covered between the objectives $[2^{i-1} c + Z, 2^i c + Z]$
shrinks exponentially in $i$ as well.
These two factors balance out exactly, giving an upper bound on the number of pivot steps
on every such segment of objectives.

\paragraph{Bound}
We make a distinction of four cases based on numbers $R > r > 0$, one very large and one very small, and some $\rho \in (0,1/2]$ and $k \in \mathbb{N}$ to be chosen later.
We write $\pi_{c,Z} : \R^d \to \linsp(c,Z)$ for the orthogonal projection onto $\linsp(c,Z)$.
\begin{itemize}
    \item The total number of bases $I \in \binom{[n]}{d}$ with $\norm{\pi_{c,Z}(x^I)} > R$
        is bounded per \Cref{lem:no-large-norms}.
    \item The total number of bases $I \in \binom{[n]}{d}$ with $\norm{\pi_{c,Z}(x^I)} < r$
        is bounded per \Cref{lem:no-small-projected-norms}.
    \item The number of bases $I$ on the path from $Z$ to $c$ satisfying
        $\norm{\pi_{c,Z}(x^I)} \in [r,R]$ and which have at least one neighbor $J$ on this path
        at distance $\norm{\pi_{c,Z}(x^I) - \pi_{c,Z}(x^J)} \geq \rho\norm{\pi_{c,Z}(x^I)}$
        is at most $O(\rho^{-1}\log(R/r))$
        as shown in \Cref{lem:donut-argument}.
    \item The bases $I$ on the path from $Z$ to $2^k c + Z$ with only close-by neighbors
        $\norm{\pi_{c,Z}(x^I) - \pi_{c,Z}(x^J)} < \rho\norm{\pi_{c,Z}(x^I)}$
        are counted in \Cref{lem:gap-multipliers-implies-separation} and \Cref{lem:long-shadow-path}
        and there are at most $O(\rho d^{5/3} k \log(n)^{3/2})$ in expectation.
\end{itemize}
The third and fourth case correspond to the vertices with long edge lengths and the vertices with large exterior angle in the proof sketch above.

\medskip
\noindent
Algorithmically, we use the 2-phase shadow vertex simplex method proposed by Vershynin \cite{ver09}.
An adaptation is required, since Vershynin requires that the objective vectors can be fixed and need not be randomly chosen.
In order to adapt his approach, we introduce the following auxiliary LPs.

\subsubsection{Auxiliary LPs}
For the smoothed constraint data $A, b$, we start out by sampling
a random objective vector $Z \in \R^d\setminus\{0\}$
from a spherically symmetric distribution
and solving $\max Z^\T x \text{ s.t. } Ax \leq 1$.
We solve this first auxiliary LP by adding
$d$ artificial constraints to create a starting vertex and traversing
a semi-random shadow path from that starting vertex to a vertex optimized by the random objective.
A number of repeated trials will be required in order to have the artificial constraints
not cut off the optimal solution to the LP, see \Cref{lem:dont-cut-off}.
At the end, this first phase results in an optimal basic feasible solution $A_I^{-1} 1_I$ to this first auxiliary LP.

In the second phase, the algorithm will operate on the feasible region of a second auxiliary LP
whose constraints are $Ax + (1-b) t \leq 1$.
We sample one more entry $Z_{d+1}$, independently from the other entries, such that the combined vector $(Z, Z_{d+1}) \in \R^{d+1}$
is spherically symmetrically distributed.
The optimal basis of the first auxiliary LP will be a set of constraints $I$,
which will be tight for an edge
of this second feasible set, and this edge connects two vertices on the combined shadow path
from $-e_{d+1}$ to $(Z,Z_{d+1})$ and onwards from $(Z,Z_{d+1})$ to $e_{d+1}$, where $e_{d+1}$ denotes the $(d+1)$th unit vector in $\R^{d+1}$ (the direction of positive $t$).
The algorithm will follow this path until it finds, on one of its traversed edges,
a feasible solution satisfying the system of linear equalities
$x^{I'} = A_{I'}^{-1} b_{I'}$ and $t=1$, with $I' \in \binom{[n]}{d}$.
This will immediately give us an optimal basic feasible solution $A_{I'}^{-1}b_{I'}$ to the linear program
$\max Z^\T x \text{ s.t. } Ax \leq b$ with random objective.
For the third and final phase, the algorithm follows the shadow path from $Z$ to $c$
on the constraints $Ax \leq b$ to find the optimal solution to the intended LP.

Hence, all shadow paths followed by our algorithm have at least one objective (start or finish)
randomly sampled from a spherically symmetric distribution that is independent of the
constraint data.
We require shadow bounds in two cases, where either the right-hand side is identical to $1$
or where the right-hand side is perturbed.

\subsubsection{Lower bound}
In \Cref{sec:lb} we construct unperturbed LP data $\bar A,\bar b, c$ for which,
when the LP data is perturbed,
any simplex path between the maximizer and minimizer of the objective $c$ will have at least
$\Omega(\frac{\sqrt{d-1}}{\sqrt{\sigma\sqrt{\ln(4/\sigma)}}})$ steps with high probability,
assuming that the number $n$ of constraints is permitted to be exponentially
large $n = \lfloor (4/\sigma)^d \rfloor$.
Although this result should not be thought of as indicative of real-world performance
due to the high number of constraints, it does demonstrate that
the dependence on $\sigma$ in \Cref{thm:main} cannot be further decreased
without increasing its dependence on $n$ from $\log(n)^{O(1)}$ to $n^{O(1/d)}$.
In that sense, the upper bound described in this paper has optimal dependence on the noise parameter $\sigma$
up to poly-logarithmic factors.

The construction involves having the rows of $\bar A$ be a set of unit vectors
that are ``well-spread-out'' on the sphere. In technical terms we require this set
to be $\sigma$-dense: for every $\theta \in \sfe$ there must be an index $i \in [n]$
such that the $i$'th row is close to $\theta$, i.e., $\norm{\theta - \bar a_i} \leq \sigma$.
Taking $\bar b = 1$, the resulting feasible region $P = \{x : Ax \leq b\}$
after perturbation will be close to the unit ball in the sense that
\[
    (1-8\sigma\sqrt{d\ln n})\ball^d \subseteq P \subseteq (1+16\sigma\sqrt{d\ln n})\ball^d
\]
with probability at least $1-n^{-d}$.
This gives the polar polytope $P^\circ = \{y \in \R^d : \sprod{y}{x} \leq 1 \;\forall x \in P\}$ a similar proximity
to the unit ball
\[
    (1-16\sigma\sqrt{d\ln n})\ball^d \subseteq P^\circ \subseteq (1+12\sigma\sqrt{d\ln n})\ball^d.
\]
Geometrically it follows, using elementary calculations, that any facet of $P^\circ$
has Euclidean diameter of at most $16\sqrt{\sigma\sqrt{d\ln n}}$.
Any simplex path on $P$ from the maximizer to the minimizer of any given linear objective function $c$
corresponds to a sequence of facets of $P^\circ$, and the length of this sequence
can be lower bounded using geometric progress along the boundary of $P^\circ$.
A similar argument based on the twin principles of roundness-induced facet diameter bounds and the connection between the primal vertex diameter and the polar facet diameter was first developed by \cite{bdghl21} in the context of random constraint matrices.
For the full lower bound argument, we refer to \Cref{sec:lb}.

\section{Preliminaries}

We write $[n] := \{1,\dots,n\}$ and $\binom{[n]}{d} := \{S \subseteq [n] : \abs{S}=d\}$.
Whenever the given dimension is clear from the context, we write $1$ for the all-ones vector
and $I$ for the identity matrix. The standard basis vectors are denoted by $e_1,\dots,e_d \in \R^d$.
Let $W \subseteq \R^d$ be a linear subspace. Then we denote the orthogonal projection onto $W$ by $\pi_W$.

The $\ell_2$-norm is $\|x\|_2 = \sqrt{\sum_{i\in[d]} x_i^2}$
and the $\ell_\infty$-norm is $\|x\|_\infty = \max_{i\in[d]} |x_i|$ for a vector $x \in \R^d$ .
A norm without a subscript is always the $\ell_2$-norm.
Further, let $\sfe$ denote the unit sphere in $\R^d$, i.e., $\sfe \coloneqq \{ x \in \R^d \colon \norm{x} =1 \}$.

For sets $A, B \subseteq \R^d$, the distance between the two is
$\dist(A, B) = \inf_{a \in A, b \in B} \|a-b\|$.
For a point $x \in \R^d$ we write $\dist(x, A) = \dist(A, x) = \dist(A, \{x\})$.
The affine hull of $d$ vectors $a_1,\dots,a_d$ is denoted as $\aff(a_i : i \in [d])$ and their
convex hull as $\conv(a_1,\dots,a_d) = \conv(a_i : i \in [d])$.

For a convex body $K \in \R^d$,  we define $\partial K \subseteq \mathrm{span}(K)$ as the boundary of $K$ in the linear subspace spanned by the vectors in $K$.

\subsection{Polytopes and Cones}

\begin{definition}[Polyhedron]
We call a convex set $Q \subset \mathbb{R}^d$ a polyhedron if it can be written as $Q = \{x \in \mathbb{R}^d : A x \leq b\}$
for some $A \in \mathbb{R}^{n \times d}$ and $b \in \mathbb{R}^n$.

\end{definition}

\begin{definition}
    Let $I \in \binom{[n]}{d}$ index a basis, let $A_I \in \R^{d \times d}$ and $b_I \in \R^d$ be the corresponding submatrix of $A$ respectively the corresponding subset of $b$ indexed by $I$ and call $x^I = A_I^{-1} b_I$ the corresponding basic solution. We say that $x^I$ and $I$ are feasible for the LP $\max c^{\top }x$ subject to $Ax \leq b$ if it satisfies $Ax^I \leq b$. We denote the set of feasible bases of the system $Ax \leq b$ by $F(A,b)$.
    We say $I$ maximizes an objective $c \in \R^d$ if $c^\T A^{-1}_I \geq 0$.
\end{definition}

\begin{definition}
    Let $\{ a_1, \dots, a_n \colon a_i \in \R^d \} $ be a set of vectors in $\R^d$. The cone $\cone(a_1, \dots, a_n)$ generated by $a_1, \dots, a_n$ is defined as $\cone(a_1, \dots, a_n) \coloneqq \{x \in \R^d \colon x = \sum_{i =1}^n \lambda_i a_i \}$ for $\lambda_i \in \R_{\geq 0}$.
\end{definition}

\subsection{Probability Distributions}\label{sec:laplace-gaussian}
All probability distributions considered in this paper will admit a probability density function with respect to the Lebesgue measure.

First we look at useful properties that probability distributions may have and which we use throughout the paper.

\begin{definition}[$L$-log-Lipschitz random variable]\label{def:log-lipschitz}
    Given $L > 0$, we say a random variable $x \in \R^d$ with probability density $\mu$ is $L$-log-Lipschitz (or $\mu$ is $L$-log-Lipschitz) if, for all $x, y\in \R^d$, we have
    \begin{align*}
        |\log(\mu(x)) - \log(\mu(y))| \leq L\|x - y\|,
    \end{align*}
    or equivalently, $\mu(x) / \mu(y) \leq \exp(L\|x - y\|)$.
\end{definition}

In the following we see an equality for the expected value of any convex function applied to any random variable.

\begin{lemma}[Jensen's inequality]\label{lem:jensen} Let $X$ be a random variable and $f$ a convex function. Then we have $f(\E[X]) \leq \E[f(X)]$.

\end{lemma}

\begin{definition}
    A random variable $X \in \R^d$ is exponentially distributed if for every $S \subseteq \R^d$ it satisfies
    $$\Pr[X \in S] = \int_{S} C e^{-\norm{x}} \dd x. $$
\end{definition}

\begin{lemma}\label{fact:norm_exp_distr}
    The normalizing constant $C$ of the exponential distribution on $\R^d$ is $C = \frac{1}{d! \vol_d(\ball^d)}$.
    For $X$ exponentially distributed on $\R^d$, the $k$'th moment of $\norm{X}$ is $\E[\norm{X}^k] = \frac{(k+d-1)!}{(d-1)!}$.
\end{lemma}
\begin{proof}
    Integrating in polar coordinates gives us the normalizing constant
    \begin{align*}
        \int_{\R^d} e^{-\norm{x}} \dd x &= \int_0^\infty \vol_{d-1}(r \sfe) e^{-r} \dd r \\
                                        &= \vol_{d-1}(\sfe) \int_0^\infty r^{d-1} e^{-r} \dd r \\
                                        &= \vol_{d-1}(\sfe) \cdot (d-1)! \\
                                        &= \vol_{d}(\ball^d) \cdot d!
    \end{align*}
    using the Gamma function. We can obtain the moments of $\norm{X}$ using a similar calculation:
    \begin{align*}
        \int_{\R^d} \norm{x}^k e^{-\norm{x}} \dd x &= \int_0^\infty \vol_{d-1}(r \sfe) r^k e^{-r} \dd r \\
                                        &= \vol_{d-1}(\sfe) \int_0^\infty r^{k+d-1} e^{-r} \dd r \\
                                        &= \vol_{d-1}(\sfe) \cdot (k+d-1)!.
    \end{align*}
    Dividing $\frac{\vol_{d-1}(\sfe) \cdot (k+d-1)!}{\vol_{d-1}(\sfe)\cdot(d-1)!} = \frac{(k+d-1)!}{(d-1)!}$
    gives the result.
\end{proof}

For the exponential distribution we have the following tail bound.

\begin{lemma}\label{lem:exponentialtail}
    Let $X$ be exponentially distributed on $\R^d$. Then  for any $t > 1$ we have
    $$\Pr[\norm{X} \geq 2ed \ln t] \leq t^{-d}.$$
\end{lemma}
\begin{proof}
    Using Markov's inequality we know for $k = d \ln n$ and $t=2ed\ln n$ that
    \begin{align*}
        \Pr[\norm{X} \geq t] &= \Pr[\norm{X}^k \geq t^k] \\
                          &\leq \frac{\E[\norm{X}^k]}{t^k} \\
                          &\leq \frac{(k+d)^k}{t^k} \\
                          &\leq \frac{(2d\ln n)^{d \ln n}}{(2ed\ln n)^{d \ln n}} = n^{-d}.
    \end{align*}
\end{proof}

The exponential distribution is 1-Lipschitz continuous.

\begin{definition}[Gaussian distribution]
    The $d$-dimensional Gaussian distribution $\mathcal{N}_d(\Bar{a}, \sigma^2I)$ with support on $\R^d$, mean $\Bar{a}\in \R^d$, and standard deviation $\sigma$, is defined by the probability density function
    \begin{align*}
        \sigma^{-d} \cdot (2\pi)^{-d/2} \cdot \exp\left(-\|s - \Bar{a}\|^2 / (2\sigma^2) \right)
    \end{align*}
    at every $s \in \R^d$.
\end{definition}

A useful standard property of the Gaussian distribution is the following tail bound:

\begin{lemma}[Gaussian tail bound]\label{lem:gaussian-tail}
    Let $x \in \R^d$ be a random vector sampled with independent Gaussian distributed entries of mean $0$ and variance $\sigma^2$.
    For any $t\geq 1$ we have
        \begin{align*}
            \Pr[\|x\| \geq t\sigma \sqrt{d}] & ~ \leq \exp(-(d/2)(t-1)^2).
        \end{align*}
        From the above, we can upper bound the maximum norm over $n$ Gaussian random vectors with mean $0$ and variance $\sigma^2$ by $4\sigma \sqrt{d \ln n}$ with probability at least $1-n^{-4d}$.
\end{lemma}

\begin{cor}[Global diameter of Gaussian random variables]\label{cor:gaussian-globaldiam}
    For any $n \geq 2$,
    let $x_1, \ldots, x_n \in \R^d$ be random variables where each $x_i$ is independent Gaussian distributed with mean $0$ and standard deviation $\sigma$. Then with probability at least $1 - n^{-d}$, $\max_{i \in [n]}\|x_i\| \leq 4 \sigma \sqrt{d\ln n}$.
\end{cor}
\begin{proof}
    From \Cref{lem:gaussian-tail}, we have for each $i \in [n]$ that
    \begin{align*}
        \Pr[\|x_i\| > 4\sigma \sqrt{d\log n}]
        \leq \exp(-\frac{d(4\sqrt{\log n} - 1)^2}{2})
        \leq \exp(-2d\log n)
        \leq n^{-1} \cdot n^{-d}.
    \end{align*}
    Then the statement follows from the union bound over $i=1,\dots,n$.
\end{proof}

\begin{theorem}[Chernoff bound, see Theorem 4.5 in \cite{probabilityandcomputing}]\label{chernoff}
   Let $X_1,\dots, X_n \in \{0,1 \}$ be $n$ independently distributed random variables. Let $X \coloneqq \sum_{i=1}^n X_i$. Then $$\Pr[X=0] \leq e^{-\E[X]/2}.$$
\end{theorem}

\begin{theorem}[Mass distribution of the sphere]\label{thm:sphere-mass}
    If $v \in \sfe$ is a fixed and arbitrary unit vector, and if $\theta \in \sfe$ is sampled uniformly at random
    from the unit sphere,
    then for any $\alpha > 0$ we have
    \[
        \Pr[\abs{\theta^\T v} \leq \alpha] \leq \alpha \sqrt{de}.
    \]
\end{theorem}
\begin{proof}
Assume that without loss of generality the fixed and arbitrary unit vector is the first standard basis vector $v = e_1$.
    We notice that for any $\alpha > 0$ the probability $
        \Pr[\abs{\theta^\T e_1} \leq \alpha] $ is given as the ratio between the volume of the unit sphere $\sfe$ intersected with the half-spaces $\{ x \in \R^d \colon x_1 \leq \alpha  \}$ and $\{ x \in \R^d \colon x_1 \geq -\alpha  \}$ and the volume of the unit sphere $\sfe$ itself.
        Further, we notice that the volume of $\sfe$ can be computed as
        \begin{align*}
            \vol(\sfe)
            = \int_{-1}^1 \vol_{d-2}( (\sqrt{1-s^2}) \mathbb{S}^{d-2}) \sqrt{ 1 + \left( \frac{-2s}{2 \sqrt{1-s^2}} \right)^2} \dd s
            = \vol_{d-2}(  \mathbb{S}^{d-2}) \int_{-1}^1 \sqrt{1-s^2}^{d-3} \dd s.
        \end{align*}
        The factor $ \left( \frac{-2s}{2 \sqrt{1-s^2}} \right)^2$ in the first equality is the derivative of the radius of the sphere $(\sqrt{1-s^2}) \mathbb{S}^{d-2}$ with respect to $s$, which came up from the surface area calculation.
        Hence, we can write the probability that the first coordinate of $\theta$ is at most $\alpha$ as
        \begin{align*}
        \Pr[\abs{\theta^\T e_1} \leq \alpha]
        = \frac{\vol_{d-2}(  \mathbb{S}^{d-2}) \int_{-\alpha}^{\alpha} \sqrt{1-s^2}^{d-3} \dd s}{\vol_{d-2}(  \mathbb{S}^{d-2}) \int_{-1}^1 \sqrt{1-s^2}^{d-3} \dd s}
        \leq \frac{ \int_{-\alpha}^{\alpha} \sqrt{1-s^2}^{d-3} \dd s}{ \int_{-1/\sqrt{d}}^{1/\sqrt{d}} .\sqrt{1-s^2}^{d-3} \dd s}.
        \end{align*}
        We will upper bound the integrand of the numerator by $1$.
        If $s \in [-1/\sqrt{d}, 1/ \sqrt{d}]$, then one calculates that $\sqrt{1-s^2}^{d-3} \in [1/\sqrt{e}, 1]$. We use this for lower bounding the denominator as
        \begin{align*}
            \Pr[\abs{\theta^\T e_1} \leq \alpha]
            \leq \frac{ \int_{-\alpha}^{\alpha} 1 \dd s}{\int_{-1/\sqrt{d}}^{1/\sqrt{d}} 1/ \sqrt{e} \dd s} = \alpha \sqrt{d e}
        \end{align*}
        and find the desired bound.
\end{proof}

\section{Algorithms}
\label{sec:algorithm}
This section will show how to adapt the algorithmic reduction of \cite{ver09} such that it can be used for the semi-random
shadow vertex method. Proofs of the stated lemmas can be found in \cite{ver09}.
The full procedure will output one of following scenarios
\begin{itemize}
    \item a vector $x \in \R^d$ with $Ax \leq 0$, certifying \emph{unboundedness},
    \item a vector $y \in \R^n$ with $y^\T A = 0$, $y \geq 0$ and $y^\T b < 0$, certifying \emph{infeasibility}, or
    \item a basis $I \in \binom{[n]}{d}$ which is both feasible $A(A_I^{-1}b_I) \leq b$ and optimal $c^\T A_I^{-1} \geq 0$.
\end{itemize}
Note that we have a rather generous definition for unboundedness: an LP can simultaneously
be unbounded and infeasible, or simultaneously be unbounded and admit an optimal basic feasible solution.
This flexibility we grant ourselves out of kindness, not out of necessity.
If one wishes to tease apart these cases explicitly, this can be done using well-known reductions.
One such option is to add a cut $c^\T x \geq -M$ for some large $M > 0$, and not perturb this constraint.
A single unperturbed constraint can be factored through the analysis in its entirety, leading to a running time blow-up of at most a factor $2$.

All that follows can be adapted to work with a more restrictive definition of unboundedness. For the sake of the clarity of our argument, we proceed with the above terminology.

\subsection{Shadow vertex method}

The  shadow vertex method \Cref{alg:shadow-vertex} is a pivot rule for the simplex method. Let two objective vectors $y, y'\in \R^d$ and a feasible basis $ I  \in \binom{[n]}{d}$ optimal for $y$ be given. We want to find a basis optimal for $y'$.
The shadow vertex pivot rule will prescribe pivot steps
in such a manner that, throughout the algorithm's duration,
the current basis is optimal for an objective $y^{\lambda} := \lambda y' + (1-\lambda)y$ with $\lambda \in [0,1]$. In the following we will describe the primal interpretation of the shadow vertex method which we summarize as \Cref{alg:shadow-vertex}.
In each iteration of \Cref{alg:shadow-vertex}, $\lambda$ is increased according to the condition of line \ref{line:max_lambda}. If $\lambda$ is found to be at least $1$, then our current basis is optimal for $y'$ and returned. Otherwise \Cref{alg:shadow-vertex} finds in lines \ref{line:find_j} -- \ref{line:new_basis} distinct indices $l$ and $j$ such that the new basis $I\setminus \{ l\} \cup \{ j\}$ is optimal for $y^{\lambda}$ and repeat with increased $\lambda$. With probability $1$, the leaving index $j$ is uniquely determined: the constraint system $Ax \leq b$ is non-degenerate since the bounds $b$ are perturbed.
Similarly, with probably $1$ the entering constraint $l$ will be uniquely determined: having one of $y,y'$ be random means we will not find dual degeneracy.
Since $y^1 = y'$, at the end of this path the method has found an optimal basis for the objective $y'$. \Cref{alg:shadow-vertex} is called the shadow vertex method because, when
the feasible set is orthogonally projected onto
the two-dimensional linear subspace $\spa(y,y')$, the vertices visited by the
algorithm project onto the boundary of the projection (“shadow”) $\pi_{\spa(y,y')}(\{x : Ax \leq b\})$ of the feasible region.
This exposition of the shadow vertex simplex method will assume certain non-degeneracy conditions without further remark, for they hold with probability $1$ under the probability model we study. See \cite{bwcachapter, huiberts2022geometric} for more detail.

\begin{algorithm}[H]

    \caption{Shadow vertex method \protect\Call{ShadowVertex}{A, b, y, y',  I }}
    \label{alg:shadow-vertex}
	\begin{algorithmic}[1]
            \State \textbf{Input:} \hspace{0.34cm} non-degenerate polyhedron $=\{x \in \R^d \colon Ax \leq b \}$
		\State \hspace{1.64cm} objective functions $y, y' \in \R^d$
                \State \hspace{1.64cm} feasible basis $I \subseteq [n]$, optimal for $y$
            \State \textbf{Output: }   basis $I \subseteq [n]$ optimal for $y'$ or \textit{unbounded} \vspace{0.2cm}

        \State $i \gets 0$ \Comment{Iteration counter}
		\State $\lambda_i \gets 0$ \Comment{Shadow progress} \label{line:after_k_change}
            \While{$\lambda_i \neq 1$} \label{lin:inner_while}
		          \State $i \gets i+1 $
                      \State $\lambda_i \gets$ supremum $\lambda$ such that $(y^{\lambda})^{\top}A_I^{-1} \geq 0$ \Comment{Maximal $\lambda$ such that $I$ is optimal for $\lambda y' + (1-\lambda)y$} \label{line:max_lambda}
                \If{$\lambda_i \geq 1$}
                    \State \textbf{return} I \Comment{If basis is optimal for $y$, return said basis}
                \EndIf
                \State $j \gets j \in I$ such that $(y^{{\lambda}^{\top}}A_I^{-1})_j = 0$ \Comment{Pivot rule. Will be unique for generic $y$ or $y'$} \label{line:find_j}
                \State $x^I \gets A_I^{-1}b_I$
                \State $s_i \gets $ supremum over all $s$ such that $A(x^I-sA_I^{-1}e_j) \leq b$ \Comment{Find simplex step length $s$}
                \If{$s_i = \infty$ }
                    \State \textbf{return} \textit{unbounded}
                \EndIf
                \State $l \gets l \in [n]\setminus I $ such that $a_l^{\top}(x^I-s_iA_I^{-1}e_j) = b_l$ \Comment{Ratio test. Will be unique}
                \State $I \gets I \setminus\{l \} \cup \{ j\}$ \label{line:new_basis}
            \EndWhile \label{line:end_inner_while}
	\end{algorithmic}
\end{algorithm}

\begin{definition}
    We denote by $\pi_{c,c'} : \R^d \to \linsp(c,c')$
    the orthogonal projection onto the span of $c$ and $c'$.
    We call the image $\pi_{c,c'}(Q)$ of a polyhedron $Q$ under $\pi_{c,c'}$ the shadow polygon.
\end{definition}
Note that, since $Q$ can be unbounded, the shadow polygon $\pi_{c,c'}(Q)$ might be unbounded.

\begin{definition}
    Given a basis $I \in \binom{[n]}{d}$ we write the corresponding solution as $x^I = A_I^{-1}b_I$.
    The set $F(A,b) \subseteq \binom{[n]}{d}$ consists of all feasible bases, i.e., bases for which $Ax^I \leq b$.

    For linearly independent $c, c' \in \R^d$, the subset $P(A,b,c,c') \subseteq F(A,b)$ is called the shadow path from $c$ to $c'$,
	and consists of all bases such that $x^I$ is maximized by some $y \in [c,c']$, i.e., for which $[c,c'] \cap A_I \R^d_{\geq 0} \neq \emptyset$.
The vertices $v_1,v_2$ on $P(A,b,c,c')$ maximizing $c$ or $c'$ are called endpoints.
\end{definition}

\begin{definition}\label{def:neighbour}
    Let $I,I' \in P(A,b,c,c')$. We say that $I'$ is a neighbor of $I$ on the shadow path if there exists an edge on the shadow polygon between $\pi_{c,c'}(x^I)$ and $\pi_{c,c'}(x^{I'})$.
Let $N(A,b,c,c',I)$ denote the set of neighbors of $I$ on the shadow path $P(A,b,c,c')$.
\end{definition}

Note that there can be other bases $J \in P(A,b,c,c')$ which, despite having intersection $\abs{I \cap J} = d-1$, are not neighbors \emph{on the shadow path} and hence $J \notin N(A,b,c,c',I)$.

\begin{definition}
	A shadow path $P(A,b,c,c') \subseteq F(A,B)$ is called non-degenerate if the pre-image $\pi_{c,c'}^{-1}(\pi_{c,c'}(x^I))$ of every basic solution $x^I$ contains no feasible solution other than $x^I$.
\end{definition}

\begin{fact}[Non-degeneracy of the shadow path]
    If the matrix $A \in \R^{n \times d}$ has independent Gaussian-distributed entries,
    which are also independent of $b, c$ and $Z$,
    then the shadow path is non-degenerate with probability $1$.
\end{fact}

\begin{fact}\label{lem:shadowbasics}
    For any $A \in \R^{n\times d}, b \in \R^n, \lambda_1,\lambda_2, \lambda_3 > 0$
    and linearly independent $c,c'\in \R^d$ we have
    \begin{itemize}
        \item $P(A,b,c,c') = P(\lambda_1 A, \lambda_1 b,\lambda_2 c, \lambda_3 c')$
        \item $P(A,b,c,c') = P(A,b,c',c)$.
    \end{itemize}
\end{fact}

\begin{fact}\label{lem:shadow-path-is-path}
    Let $P(A,b,c,c')$ be a non-degenerate shadow path.
    The vertices of $P(A,b,c,c')$, along with the steps taken between them in order, form a path in the graph-theoretical sense.
    If $P(A,b,c,c') \geq 2$ then for any $I \in P(A,b,c,c')$, we have $\abs{N(A,b,c,c',I)} = 2$
    except for the two endpoints where it is~$1$.
\end{fact}

\begin{fact}\label{lem:compose-paths}
    Let $A \in \R^{n \times d}, b \in \R^n$ and let $c,c' \in \R^d$ be linearly independent objectives.
    Let $P(A,b,c,c')$ be a non-degenerate shadow path.
    Then we have that $\abs{P(A,b,c,y) \cap P(A,b,y,c')} \leq 2$ for all $y \in [c,c']$.
    Moreover, if $y_1,y_2,\dots,y_k \in [c,c']$ then $\sum_{i=1}^{k-1} \abs{P(A,b,y_i, y_{i+1})} \leq \abs{P(A,b,c,c')} + 2k$.
\end{fact}

\subsection{Phase 1 and the first auxiliary LP}
In phase 1 we construct our first auxiliary LP, with a new right-hand side vector and a number of auxiliary constraints.
The algorithm will make a number of pivot steps in this phase in order to pivot out these auxiliary constraints.

In order to construct this auxiliary LP, we sample $Z \in \R^d$ with independent entries,
each following a Gaussian distribution with mean $0$ and standard deviation $1$.
The LP we will solve in this step of the algorithm is:
\begin{align*}
    \max Z^\T&x \tag{Unit LP}\label{unit-LP} \\
    Ax &\leq 1
\end{align*}
By construction, the all-zeroes solution certifies this LP to be strictly feasible. We note that, of the original LP data,
only the constraint matrix $A$ appears in \eqref{unit-LP}.
In order to obtain a feasible starting basis that is independent of the noise on $A$,
we will add $d$ artificial constraints.
Let $\bar s_1,\dots,\bar s_d \in \R^d$ be such that $\conv(\bar s_1,\dots,\bar s_d)$ is a regular $(d-1)$-dimensional simplex,
and furthermore satisfy $e_d^\T \bar s_i = 3$ and $\norm{3e_d - \bar s_i} = \frac{1}{10\sqrt{\ln d}}$
for each $i=1,\dots,d$.
Sample independently perturbed vectors $s_1,\dots,s_d \in \R^d$ with means respectively equal to $\bar s_1,\dots,\bar s_d$
and standard deviation $\sigma > 0$.
Let $R \in \R^{d \times d}$ denote a uniformly random rotation matrix following the Haar measure on the orthogonal group $O(d)$ and construct \eqref{unit-LP'} as follows:
\begin{align*}
    \max Z^\T&x \tag{Unit LP'} \label{unit-LP'}\\
    Ax &\leq \mathbf{1} \\
    (Rs_i)^\T x &\leq \mathbf{1}\qquad \forall i=1,\dots,d.
\end{align*}
We take this construction from \cite{ver09} who shows the following helpful properties:

\begin{lemma}\label{lem:dont-cut-off}
    Suppose $0 < \sigma < \frac{1}{6\sqrt{d\log n} + d^{3/2}\log d}$.
    If \eqref{unit-LP} admits an optimal solution $x^*$ then with probability at least $0.25$
    it satisfies $(Rs_i)^\T x^* \leq 0$ for all $i=1,\dots,d$.
    This probability is independent of $A$.
\end{lemma}

\begin{lemma}\label{lem:construct-starting-basis}
    Let $S \in \R^{d\times d}$ denote the matrix with rows $s_1,\dots,s_d$.
    Conditional on the rows of $A$ each having norm at most $2$ then,
    with probability at least $0.9$, independent of $A$, the basic solution $(RS)^{-1} \mathbf{1}$ is feasible
    and satisfies $(Re_d)^\T(RS)^{-1} \geq 0$.
\end{lemma}

The outcome of these lemmas is as follows. We can construct \eqref{unit-LP'}
as described, take $Re_d$ as our fixed objective and $Z$ as our random objective,
and attempt to follow the shadow path starting at the constructed basis from fixed objective $Re_d$ to random objective $Z$.
With constant probability this succeeds, and with constant probability
this gives an optimal basic feasible solution to \eqref{unit-LP}.
On a failure, the procedure is repeated with fresh samples $R, S$ until success.
Since the success probability can be made independent of $A$, the lengths of
all attempted shadow paths are identically distributed.
This follows almost exactly as first described by \cite{ver09}.
We prove in \Cref{thm:fixedrhs} that these paths have expected length $O(\sqrt{\sigma^{-1}\sqrt{d^{11} \log(n)^7}})$.

Since the smoothening of the system should not interfere with the artificial constraints forming a basic feasible solution, one needs to restrict the perturbation size $\sigma$.
Dadush and Huiberts \cite{DH18} claimed that the restriction of the perturbation size $\sigma$ suffices $\sigma \leq \frac{c}{\max\{ \sqrt{d \log n}, \sqrt{d} \log d \}}$ for some constant $c > 0$.
Once again this is without loss of generality if one is willing to accept a constant additive factor
independent of $\sigma$.

If any attempted shadow path finds that the feasible region of \eqref{unit-LP'} is unbounded,
then according to our definition so is the original LP. Thus we may simply return \emph{unbounded}
whenever this occurs.

Phase 1 ends with having found an optimal basic feasible solution to \eqref{unit-LP} with the random objective function.
Having found an optimal solution to \eqref{unit-LP} for the random objective in phase 1, we can use it as input to the second auxiliary LP \eqref{int-LP} explained in the next section.

\subsection{Phase 2: a second auxiliary LP}
In phase 2 we use the previously found basis.
The pivot steps in this phase will monotonically improve the feasibility of the basis,
in order to obtain a basic feasible solution on \eqref{input-LP}.

Sample a $(d+1)$th coordinate $Z_{d+1} \in \R$ for $Z \in \R^d$ such that $Z_{d+1} $ is standard Gaussian distributed random variable.
This makes the extended vector $(Z, Z_{d+1})$ follow a spherically symmetric probability distribution.
We now think of bases $I \in \binom{[n]}{d}$ to \eqref{unit-LP}
as indexing edges in the \emph{interpolation LP} which has constraints
\begin{align*}
    Ax + (\mathbf{1}-b)t &\leq \mathbf{1}.
	\tag{Int-LP}\label{int-LP}
\end{align*}
On this LP we will consider 3 different objectives: we either minimize~$t$,
maximize $Z^\T x + Z_{d+1}t$, or maximize~$t$.

The slice of \eqref{int-LP} where $t=0$ equals the feasible region of \eqref{unit-LP},
meaning that the optimal basis $I$ from phase 1 indexes a set of constraints that is tight for some edge
of \eqref{int-LP} that passes through the $t=0$ slice.
Both endpoints of this edge are part of the combined shadow path
\[
    P\Big( \big(A,(1-b)\big), 1, -e_{d+1}, (Z,Z_{d+1})\Big)
        \cup
    P\Big( \big(A,(1-b)\big), 1, (Z,Z_{d+1}), e_{d+1}\Big).
\]
As such, the second phase can be started somewhere on this path
and we are able to use the shadow vertex method to follow the combined shadow path in order to increase $t$.
The slice of \eqref{int-LP} where $t=1$ has a feasible region equal to the original LP, meaning that, as soon as we find an edge crossing this slice, we have obtained a basic feasible solution to start phase 3.
Again by \Cref{thm:fixedrhs} we know that this path has length $O(\sqrt{\sigma^{-1}\sqrt{d^{11} \log(n)^7}})$.

If the shadow vertex method stops early, finding that the optimal solution to
\begin{align*}
    \max \ t \\
    Ax + (\mathbf{1}-b)t &\leq \mathbf{1}.
\end{align*}
has value strictly less than $1$, then this optimal basic feasible solution has a corresponding dual solution that functions as a certificate that the feasible set $\{x : Ax \leq b\}$
is empty. In that case the algorithm may return said certificate.

\subsection{Phase 3: the input LP}
When phase 2 finds an edge of \eqref{int-LP} that crossed the $t=1$ slice,
its tight constraints give a basic feasible solution $A_I^{-1}b_I$ to \eqref{input-LP}.
Moreover, due to properties of the shadow vertex simplex method
this basic feasible solution is optimal for the random objective $\max Z^\T x$.
See \cite{bwcachapter, huiberts2022geometric} for more details.

Thus, in phase 3 all that remains is to follow the semi-random shadow path from $Z$ to $c$.
We prove in \Cref{thm:smoothedrhs} that this can be done using an expected
$O(\sqrt{\sigma^{-1}\sqrt{d^{11} \log(n)^7}})$ pivot steps.
In this third phase, all pivot steps will monotonically increase the objective value of the basis.
This finishes the algorithmic reduction.

\section{Semi-random shadow bound}
\label{sec:Semi-random_shadow_bound}

We will prove a semi-random shadow bound in two cases: either when $b$ is perturbed as is prescribed for smoothed analysis,
or when $b$ is fixed to be the all-ones vector.

Although for algorithmic purposes we were satisfied with any rotationally symmetric distribution for $Z$,
the proofs in this section will have the norm $\norm{Z}$ require a specific distribution as well.
For that purpose, recall from \Cref{lem:shadowbasics} that for any $A \in \R^{n\times d}, b \in \R^n, \lambda_1,\lambda_2 > 0$
and linearly independent $c,c'\in \R^d$ we have
\[
    P(A,b,c,c') = P(A,b,\lambda_1 c, \lambda_2 c').
\]
As such, changing the norm of the random vector $Z$ has no consequences for the simplex path taken.
We will sample $Z$ to be a $1$-log-Lipschitz random variable as per \Cref{def:log-lipschitz}.

\subsection{Pivot steps close to the fixed objective}\label{sub:close-to-fixed}

As the algorithm traverses the shadow path, the main analysis requires there to be a ``large amount'' of randomness
in the objectives that are visited. This is true for the majority of the path, except when the angle between
the ``current objective'' and the LP's true objective is small.
For that reason, we must treat this part of the shadow path separately first.
The following statement is inspired by the angle bound of \cite{ST04},
but to keep this present document self-contained we give a simple proof of a similar but much weaker result.
For our purposes this weaker version suffices.
We refer to the discussion section for the possibility of strengthening this section's analysis.

\begin{definition}(Angle)
    Given two nonzero vectors $s,s'\in\R^d$, the angle
    $\angle(s,s') \in [0,\pi]$ between $s$ and $s'$ is defined to be the unique number such that
    $\cos(\angle(s,s'))\cdot\norm{s}\cdot\norm{s'} = s^\T s'$.

    For two sets $S,S' \subset \R^d$ we define
	\begin{align*}
	    \angle(S,S') = \inf_{\substack{s \in S\setminus\{0\}\\ s' \in S'\setminus\{0\}}}\angle(s,s').
	\end{align*}
\end{definition}

\begin{lemma}[Angle bound]\label{lem:replacement-angle-bound}
    Let $c \in \R^d \setminus\{0\}$ be an objective vector.
    Assume that $a_1,\dots,a_n \in \R^d$ are independent Gaussian distributed random vectors, each with standard deviation
    $\sigma \leq 1/4\sqrt{d\ln n}$ and $\norm{\E[a_i]} \leq 1$.
    Let $0 < \eps \leq \pi/10$.
    Then
    \[
        \Pr\Big[ \exists J \in \binom{[n]}{d-1} : \angle\big(c, \cone(a_j : j \in J)\big) < \eps\Big]
        \leq 4d \cdot n^d \cdot \frac{\eps}{\sigma\sqrt{2\pi}} + n^{-d}.
    \]
\end{lemma}
\begin{proof}
    Consider the event $E$ that, for every $J \in \binom{[n]}{d-1}$ and every $j \in J$,
    we have $\dist(a_j, \linsp(\{c\} \cup \{a_{i} : i \in J\setminus\{j\}\})) \geq 2d\eps$.
    Moreover, consider the event $D$ that, for every $j \in [n]$, we have $\norm{a_j} \leq 2$.
    We first show that their intersection $E \wedge D$ implies that for all $J \in \binom{[n]}{d-1}$ we have $\angle(c, \cone(a_i : i \in J)) \geq \eps$.
    After that we will show that $\Pr[\neg(E \wedge D)] \leq n^d \cdot \frac{4 d \eps}{\sigma\sqrt{2\pi}} + n^{-d}$.

    Assume that $E$ and $D$ hold. Let $J \in \binom{[n]}{d-1}$ be arbitrary.
    By our assumption of $E$, for each $j \in J$ there exists some separator
    $y_j \in \linsp(\{c\} \cup \{a_i : i \in J \setminus\{j\}\})^\perp$
    with $\norm{y_j} = 1$ that certifies this distance through the inequalities $y_j^\T a_j \geq 2d\eps$
    and $y_j^\T c = 0$ and $y_j^\T a_i = 0$ for each $i \in J\setminus\{j\}$.
    For their sum $y = \sum_{j \in J} y_j$ we know that
    $y^\T a_j \geq 2d\eps$ for all $j \in J$, as well as that $y^\T c = 0$.
    Now consider any $p \in \cone(a_j : j \in J)$ that achieves $\angle(c,p) = \angle(c,\cone(a_j : j \in J))$.
    Without loss of generality we assume $p \in \conv(a_j : j \in J)$.
    In particular we know from the above that $y^\T p \geq 2d\eps$.
    The triangle inequality gives us that $\norm{y} \leq \sum_{j \in J} \norm{y_j} = d$. We further deduce from $D$ that $\norm{p} \leq \max_{j \in J} \norm{a_j} \leq 2$.
    From the definition of angle we get $\cos(\angle(y,p)) = y^\T p \cdot \norm{y}^{-1} \cdot \norm{p}^{-1} \geq \eps$.
    In particular, we find that $\angle(y,p) \leq \pi/2 - \eps$.
    We know that $\angle(c,y) = \pi/2$ due to $y^\T c = 0$, and hence the triangle inequality on the sphere
    gives us $\pi/2 = \angle(c,y) \leq \angle(c,p) + \angle(y,p) \leq \pi/2 - \eps + \angle(c,p)$.
    We can rearrange this to get $\angle(c,p) \geq \eps$.

    It remains to show that $\Pr[\neg(E \wedge D)] \leq n^d \cdot \frac{4 d\eps}{\sigma\sqrt{2\pi}} + n^{-d}$.
    We use the union bound:
    \begin{align}
        \Pr[\neg(E \wedge D)] &\leq \Pr[\neg E] + \Pr[\neg D] \nonumber \\
                              &\leq \Pr[\neg D] +  \sum_{J \in \binom{[n]}{d-1}} \sum_{j \in J}
                              \Pr[\dist(a_j, \linsp(\{c\} \cup \{a_i : i \in J \setminus\{j\}\}) \leq 2d\eps].\label{eq:angleunionbound}
    \end{align}
    Since $\sigma \leq 1/(4 \sqrt{d\ln n})$ and $\norm{\E[a_i]} \leq 1$ we know that $\norm{a_i} > 2$ implies
    $\norm{a_i - \E[a_i]} > 4\sigma\sqrt{d\ln n}$, so
    \Cref{cor:gaussian-globaldiam} gives that $\Pr[\neg D] \leq n^{-d}$.
	The double summation has $\binom{n}{d-1}\cdot (d-1) \leq n^d$ terms in total,
    so in the remainder we will upper bound the summand uniformly over all choices of $J$ and $j$.
    For that purpose, let $J \in \binom{[n]}{d-1}$ and $j \in J$ be arbitrary.
    From the principle of deferred decision, we may consider $V := \linsp(\{c\} \cup \{a_i : i \in J\setminus\{j\}\})$ to be fixed.
    Write $y_j \in \sfe$ to be one of the two unit normal vectors to this linear subspace $V$.
    We are interested in the distance $\dist(a_j, V) = \abs{y_j^\T a_j}$.

    Note that $V$ depends only on the values of $a_i$ for $i \in J \setminus\{j\}$, and as such $y_j$
    is independent of $a_j$.
    That makes the inner product $y_j^\T a_j$ follow a Gaussian distribution with mean $y_j^\T \E[a_j]$
    and standard deviation $\sigma$. The probability density function of this random variable is uniformly upper bounded
    by $\frac{1}{\sigma\sqrt{2\pi}}$, and hence the probability that it is contained in an interval of length $2d\eps$ is at most
    \[
        \Pr[\dist(a_j, V) < 2d\eps] = \Pr\left[y_j^\T a_j \in (-2d\eps, 2d\eps)\right] \leq \frac{4d\eps}{\sigma\sqrt{2\pi}}.
    \]
    Bounding the terms of \eqref{eq:angleunionbound} as described above closes out the proof.
\end{proof}

In order to upper bound the number of pivot steps between objectives with small angle between them on the total shadow path, we need a slightly different characterization,
captured by the following lemma.
\begin{lemma}
    \label{cor:small-angle-pivot-count}
    Let $c \in \sfe$ be a fixed objective, and let $Z \in \R^d$ be a random objective that
    is linearly independent of $c$ and satisfies $\Pr[\norm{Z} \geq t] \leq n^{-d}$ for some $t>1$.
    Assume $b \in \R^n$ is arbitrary, and that $a_1,\dots,a_n$ are independent Gaussian distributed
    random vectors each with standard deviation $n^{-2d} \leq \sigma \leq 1/4\sqrt{d \ln n}$ and $\norm{\E[a_i]} \leq 1$.
Write $k = 5 d \lceil \log_2(nt)\rceil$.
   The expected length of the shadow path between the objective $c$ and the perturbed objective $2^k c + Z$ satisfies
    \[
        \E\big[\abs{P(A,b,2^k c + Z,c)}\big] \leq 7.
    \]
\end{lemma}
\begin{proof}
    In the event that $\norm{Z} \geq t$ we count at most $\binom{n}{d}$ distinct bases.
    Since $\Pr[\norm{Z} \geq t] \leq n^{-d}$, the expected number of pivot steps
    incurred by this situation is at most $1$.
    For that reason, we will for the remainder of this proof only consider the case $\norm{Z} < t$.

    We calculate $2^k \geq t \cdot n^{5d} \geq \frac{2dt \cdot n^{2d}}\sigma$.
    Note that $\angle(2^kc + Z, c) = \angle(2^k c + Z, 2^k c)$ and
    consider the triangle $\triangle(0, 2^k c + Z, 2^k c)$.
    Let us abbreviate its vertices by $a = 2^k c + Z$ and $a' = 2^k c$, so that our triangle is $\triangle(0,a,a')$.
    That $c$ lies on the unit sphere together with the above lower bound on $2^k$ gives that $\norm{2^k c} \geq \frac{2d t n^{2d}}{\sigma} > 2\norm{Z}$.
    Using the triangle inequality $\norm{a-a'} = \norm{Z} \leq \norm{2^k c} - \norm{Z} \leq \min(\norm{a},\norm{a'})$ we find that the edge $[a,a']$
    has the shortest length of the three sides of the triangle.
    The shortest side is opposite the smallest angle and hence $\angle(a, a') \leq \pi/3$.

    We recall the law of sines to derive
    \[
        \frac{\sin(\angle(a,a'))}{\norm{a-a'}}
        = \frac{\sin(\angle(0-a, a'-a))}{\norm{a'}}
        \leq \frac{1}{\norm{a'}} = \frac{1}{2^k} \leq \frac{\sigma\sqrt{2\pi}}{2dt\cdot n^{2d}}
	\leq \frac{\sigma\sqrt{2\pi}}{2d\cdot\norm{Z}\cdot n^{2d}}.
    \]
    This gives an upper bound on the sine of our desired angle because $\norm{a-a'} = \norm{Z}$. To relate this to the angle itself,
    recall that $\angle(a, a') \leq \pi/3$.
    For any angle $\alpha$ with value in $ [0, \pi/3]$ one has $\sin(\alpha) > 0.8 \alpha$,
    so in particular
    \[
        \angle(a,a') \leq \frac{5}{4} \sin(\angle(a,a')) \leq \frac{5\sigma\sqrt{2\pi}}{4\cdot 2d \cdot n^{2d}}.
    \]
    As such, our two objectives $2^k c + Z$ and $c$ have an angle at most $\frac{5\sigma\sqrt{2\pi}}{4 \cdot 2d \cdot n^{2d}}$
    between them.
    If $\abs{P(A,b,2^k c + Z, c)} \geq 2$, i.e., if there was a pivot step taken between the two objectives, then
    that implies there is a basis $I \in P(A,b,2^k c + Z, c)$ such that
    $A_I^{-\T} c \geq 0$ but $A_I^{-\T}(2^k c + Z) \ngeq 0$.
    The line segment $[2^k+Z, c]$ intersects both the normal cone for $I$ and its complement, hence it intersects the boundary of the cone.
    This implies that there is a subset $J \subset I, \abs{J} = d-1$
    and a point $p \in [2^k c + Z, c] \cap \linsp(a_j : j \in J)$.
    This point must satisfy $\angle(p,c) \leq \angle(2^k c + Z, c) \leq \frac{5\sigma\sqrt{2\pi}}{8d\cdot n^{2d}}$,
    implying that in fact $\angle(c, \linsp(a_j : j \in J)) \leq \frac{5\sigma\sqrt{2\pi}}{8d\cdot n^{2d}}$.

    \Cref{lem:replacement-angle-bound} shows us that the probability of this happening is at most $6 n^{-d}$.
    Counting at most $\binom{n}{d}$ pivot steps in this case, we may conclude
    \[
        \E[\abs{P(A,b,2^k c + Z, c)}] \leq \binom{n}{d} \Pr[\norm{Z} \geq t] + \binom{n}{d}\Pr[\abs{P(A,b,2^k c + Z, c)} > 1 \mid \norm{Z} < t] \leq 7.
    \]
\end{proof}

\subsection{Multipliers}\label{sec:multipliers}
We will give a bound on the expected number of pivot steps for most of the shadow path: the path segment $P(A,b, Z, 2^k c +Z)$.
To start, we require the following theorem proven by \cite{bbhk}.

\begin{theorem}[\cite{bbhk}] \label{thm:segment-intersects-middle-of-cone}
    Let $B \in \R^{d \times d}$ be an invertible matrix, every whose column has
    Euclidean norm at most $2$,
    and define, for any $m \geq 0$, $C_m = \{ x \in \R^d : B^{-1}x \geq m \mathbf{1} \}$.
    Suppose $c, c'  \in \R^d$ are fixed.
    Let $Z \in \R^d$ be a random vector with $1$-log-Lipschitz probability density $\mu$.
    Then
    \[
        \Pr\left[ [c + Z, c' +  Z] \cap C_m\neq\emptyset\right]
            \geq 0.99 \Pr\left[ [c + Z, c' +  Z] \cap C_0\neq\emptyset\right]
    \]
    for $m = \ln(1/0.99)/2d$.
\end{theorem}

The elements of the shadow path satisfying the property described above form a set that we will keep track of
through the following definition.
\begin{definition}
    Given $A \in \R^{n\times d}$ and $c,c' \in \R^d$, and a threshold $m > 0$,
	the set of bases \emph{with good multipliers} is
    \[
        M(A,c,c',m) = \Big\{I \in \binom{[n]}{d} \;\Big|\; \exists y \in [c,c'] \quad\text{s.t.}\quad y^\T A_I^{-1} \geq m \mathbf{1}\;\Big\}.
    \]
\end{definition}

In the language of this definition, the previous theorem says that most bases
with all nonnegative multipliers $I \in M(A,c+Z,c'+Z,0)$ will have good multipliers $I \in M(A,c+Z,c'+Z,\ln(1/0.99)/2d)$.
We make this explicit with the following corollary.
\begin{cor}\label{cor:multipliers}
    For any fixed $A \in \R^{n \times d}$ with rows of norm at most $2$ and any fixed $c,c' \in \R^d$,
    if $Z \in \R^d$ has a $1$-log-Lipschitz probability density function then for $m=\ln(1/0.99)/2d$ we have
    \[
        \Pr\big[I \in M(A,c+Z,c'+Z,m)\big] \geq 0.99 \cdot \Pr\big[ I \in M(A,c+Z,c'+Z,0)\big].
    \]
\end{cor}
\begin{proof}
    Write $C_m = \{ y \in \R^d : A^{-\T}y \geq m \mathbf{1} \}$. We observe that equivalent events have equal probability and hence
    \[
          \Pr\big[I \in M(A,c+Z,c'+Z,m)\big]
              = \Pr\big[ \exists y \in [c+Z, c'+Z] \colon A_I^{-\T}y \geq m \mathbf{1} \big]
              = \Pr\big[ [c+Z, c'+Z] \cap C_m \neq \emptyset  \big],
    \]
    and similarly $\Pr[M(A,c+Z,c'+Z,0)] = \Pr\big[ [c+Z,c'+Z] \cap C_0 \neq \emptyset \big]$.
    At this point we can directly apply \Cref{thm:segment-intersects-middle-of-cone} to the invertible matrix $A_I^\T$ and get
    \begin{align*}
        \Pr\big[I \in M(A,c+Z,c'+Z,m)\big]
            &= \Pr\big[ [c+Z,c'+Z] \cap C_m \neq \emptyset  \big] \\
            &\geq 0.99 \Pr\big[ [c+Z, c'+Z] \cap C_0 \neq \emptyset  \big] \\
            &= \Pr\big[I \in M(A,c+Z,c'+Z,0)\big].
    \end{align*}
\end{proof}

\subsection{Slack}\label{sec:slack}
Having good multipliers alone is not sufficient, because we want every vertex on the shadow-path to be ``well-separated''
from the others.
Over the course of this subsection we will
prove that all bases which have non-negligible probability of being feasible
also have a good probability of being feasible by a good margin, i.e., the minimum non-zero slack is bounded away from $0$.
This subsection extends an argument first developed in Section 5.3 of \cite{hlz}.
We require a few facts about the Gaussian distribution.  First a technical
lemma about the range in which we may treat the Gaussian distribution as having
a log-Lipschitz probability density function.

\begin{lemma}[Gaussian as log-Lipschitz]\label{lem:gaussian-log-lipschitz}
    Assume $s \in \R$ is Gaussian distributed with variance $\sigma^2$ and denote its probability density function by $f(\cdot)$.
	If $t \in \R, p \in (0,1/e]$ and $\eps \in (0,\sigma\sqrt{\ln p^{-1}}]$ satisfy
	$\Pr[s \geq t - \eps] \geq p$ and $\Pr[s \leq t] \geq p$ then
    for any $x_1, x_2 \in [ t - 4\sigma\sqrt{\ln p^{-1}}, t + 4\sigma\sqrt{\ln p^{-1}}]$
    we have
    \[
        \frac{
            f(x_1)
        }{
            f(x_2)
        }
        \leq
        \exp\Big(8\sigma^{-1}\sqrt{\ln p^{-1}} \cdot \abs{x_1 - x_2}\Big).
    \]
\end{lemma}
\begin{proof}
    We first prove that $t \leq \E[s] +  4\sigma\sqrt{\ln p^{-1}}$. Suppose not, then we use the fact that implications between events give rise to inequalities between probabilities in order to bound
    \begin{align*}
        \Pr[s \geq t - \eps] = \Pr\Big[s - \E[s] \geq t-\E[s] - \eps\Big]
                      \leq \Pr\Big[\abs{s - \E[s]} > 3\sigma\sqrt{\ln p^{-1}}\Big]
    \end{align*}
    Since $\sqrt{\ln p^{-1}} \geq 1$, we conclude from \Cref{lem:gaussian-tail}
    that this last probability is strictly less than $p$, giving a contradiction.
    Similarly, we may prove that $t \geq \E[s] - 4\sigma\sqrt{\ln p^{-1}}$ by assuming its opposite
    and computing
    \begin{align*}
        \Pr[s \leq t] = \Pr\Big[\E[s] - s \geq \E[s] - t\Big]
                      \leq \Pr\Big[\abs{s - \E[s]} \geq 4\sigma\sqrt{\ln p^{-1}}\Big],
    \end{align*}
    once again leading to a contradiction by way of \Cref{lem:gaussian-tail}.

    Recall that the probability density function of $s$ is given by
    $f(x) = \frac{1}{\sigma\sqrt{2\pi}} e^{\frac{-(x-\E[s])^2}{2\sigma^2}}$,
    which means that on the interval
    $x_1,x_2 \in \big[t - 4\sigma\sqrt{\ln p^{-1}}, t + 4\sigma\sqrt{\ln p^{-1}}\big]
    \subseteq \big[\E[s] - 8\sigma\sqrt{\ln p^{-1}}, \E[s] + 8\sigma\sqrt{\ln p^{-1}}\big]$
    it satisfies
    \begin{align*}
        \log(f(x_1)) - \log(f(x_2)) &= \frac{1}{2\sigma^2}\left((x_2 - \E[s])^2 - (x_1 - \E[s])^2\right) \\
                                    &= \frac{1}{2\sigma^2}\left(x_2^2 - 2x_2\E[s] - x_1^2 + 2x_1\E[s]\right) \\
                                    &= \frac{1}{2\sigma^2}\cdot(x_1 + x_2 - 2\E[s]) \cdot (x_2 - x_1) \\
                                    &\leq 8\sigma^{-1}\sqrt{\ln p^{-1}} \cdot \abs{x_1 - x_2}.
    \end{align*}
    This is equivalent to our desired statement.
\end{proof}

With the above lemma in place,
we can set out to prove the main statement that we require of the Gaussian distribution.
\begin{lemma}[General condition-reversing interval lemma]\label{lem:general-condition-reversing}
    Let $L > 0$ be arbitrary.
    Suppose $s \in \R$ is a continuous random variable whose probability density function $f : \R \to \R_{\geq 0}$ satisfies the following limited log-Lipschitz property:
    For every $x_1, x_2 \in [t-2/L,t+2/L]$ we have $f(x_1)/f(x_2) \leq e^{L\abs{x_1-x_2}}$.
    Then for any $\eps \in [0,1/L]$ we have
    \[
        \Pr[s \geq t-\eps \mid s \leq t] \leq 31 \eps L \cdot \Pr[s \geq t].
    \]
\end{lemma}
\begin{proof}
    We start by proving that
    $\Pr\big[s \in [t-\eps, t]\big] \leq e^2 \eps L \cdot \Pr\big[s \in [t, t+1/L]\big]$.
    In the final paragraphs of this proof we will extend this statement into the desired conclusion.

    Since $\eps \leq 1/L$ we are within the log-Lipschitzness range to
    bound our intended left-hand side as
    \[
        \Pr\big[s \in [t-\eps, t]\big]
        = \int_{t-\eps}^t f(x) \dd x
        \leq \int_{t-\eps}^t f(t) e^{L\abs{x-t}} \dd x
        \leq e \eps f(t).
    \]
    Similarly, we may use this log-Lipschitzness property to lower bound the probability in our intended right-hand side and find
    \begin{align*}
        \Pr\big[ s \in [t, t + 1/L]\big]
                      &= \int_t^{t + 1/L} f(x) \dd x \\
                      &\geq \int_t^{t + 1/L} f(t) e^{-L\cdot\abs{x-t}} \dd x \\
                      &\geq e^{-1} f(t) L^{-1}.
    \end{align*}
    Chaining these two inequalities together, we find
    \begin{equation}\label{eq:numeratorbound}
        \Pr\big[s \in [t-\eps, t]\big] \leq e^2 \eps L \Pr\big[s \in [t,t+1/L] \big].
    \end{equation}
    This is the initial statement mentioned at the start of this proof.
    We now define the affine transformation $T : \R \to \R$ to satisfy
    $T(t-2/L) = t-2/L$ and $T(t+1/L) = t$. That is, we have $T(x)=\frac{2}{3}x + \frac{1}{3}(t-2/L)$.
    For this transformation we observe that the Jacobian is $\frac{2}{3}$,
    and for any point $y \in [t-2/L, t+1/L]$ we have $\abs{y-T(y)} \leq 1/L$.
    We use the transformation to give a change of variables and find
    \begin{align}
        \Pr[s \leq t] &= \int_{-\infty}^t f(x) \dd x \nonumber\\
                      &= \int_{-\infty}^{t-2/L} f(x) \dd x + \int_{t-2/L}^t f(x) \dd x \nonumber\\
                      &= \int_{-\infty}^{t-2/L} f(x) \dd x +
                      \frac{2}{3}\int_{t-2/L}^{t+1/L} f( T(y)   ) \dd y \nonumber\\
                      &\geq \int_{-\infty}^{t-2/L} f(x) \dd x + \frac{2}{3e}\int_{t-2/L}^{t+1/L} f(y) \dd y \nonumber\\
                      &\geq \frac{2}{3e} \int_{-\infty}^{t+1/L} f(y) \dd y \nonumber\\
                      &= \frac{2}{3e} \Pr[s \leq t+1/L] \label{eq:denominatorbound}
    \end{align}
    Using \eqref{eq:numeratorbound} and \eqref{eq:denominatorbound} in order to bound the numerator and the denominator, we
    can now prove the lemma as follows
    \begin{align*}
        \Pr[s \geq t-\eps \mid s \leq t] &= \frac{\Pr\big[s \in [t-\eps, t]\big]}{\Pr[s \leq t]} \\
                                         &\leq \frac{3e\Pr\big[s \in [t-\eps, t]\big]}{2\Pr[s \leq t+1/L]} \tag{by \eqref{eq:denominatorbound}} \\
                                         &\leq \frac{3e^3 \eps L \cdot \Pr\big[s \in [t, t+1/L]\big]}{2\Pr[s \leq t+1/L]} \tag{by \eqref{eq:numeratorbound}} \\
                                         &\leq \frac{31 \eps L \cdot \Pr\big[s \in [t, t+1/L]\big]}{\Pr[s \leq t+1/L]} \\
                                         &= 31\eps L \cdot \Pr[s \geq t \mid s \leq t + 1/L].
    \end{align*}
    In order to establish our final inequality, we use $\Pr[s \geq t \mid s \leq t+1/L] \leq 1$ to find
	\begin{align*}
		\Pr[s \geq t] &= \Pr[s > t+1/L] + \Pr[s \geq t \mid s \leq t+1/L] \Pr[s \leq t+1/L]\\
		&\geq \Pr[s \geq t \mid s \leq t+1/L]\cdot\Big(\Pr[s > t+1/L] + \Pr[s \leq t+1/L]\Big)\\
		&= \Pr[s \geq t \mid s \leq t+1/L].
	\end{align*}
    These two inequalities, then, prove the lemma as
    \[
    \Pr[s \geq t-\eps \mid s\leq t] \leq 31\eps L\cdot \Pr[s \geq t \mid s \leq t + 1/L] \leq 31\eps L \Pr[s \geq t]. \qedhere
    \]
\end{proof}

\begin{lemma}[Gaussian condition-reversing interval lemma]\label{lem:gaussian-condition-reversing}
    Suppose $s \in \R$ is Gaussian distributed with variance $\sigma^2$.
    For
    $t \in \R, p \in (0,1/10)$, write
    $L = 8\sigma^{-1}\sqrt{\ln p^{-1}}$ and pick any $0 \leq \eps \leq 1/L$.
    Assuming that $\Pr[s \geq t - \eps] \geq p$ and $\Pr[s \leq t] \geq p$ we have
    \[
        \Pr[s \geq t-\eps \mid s \leq t] \leq 31 \eps L \cdot \Pr[s \geq t].
    \]
\end{lemma}
\begin{proof}
    Let $f$ denote the probability density function for $s$.
    By \Cref{lem:gaussian-log-lipschitz} and our assumption on $p$ we know that
    $f(x_1)/f(x_2) \leq \exp{(L\cdot\abs{x_1 - x_2})}$ for any two points
    $x_1, x_2 \in [t-4\sigma\sqrt{\ln p^{-1}}, t+4\sigma\sqrt{\ln p^{-1}}]$.
    In particular it holds for $x_1,x_2 \in [t-2/L, t+2/L]$
    since $2/L = \frac{\sigma}{4\sqrt{\ln p^{-1}}} \leq 4\sigma\sqrt{\ln p^{-1}}$.
    We thus satisfy the limited log-Lipschitzness assumption of \Cref{lem:general-condition-reversing}.
    We apply \Cref{lem:general-condition-reversing} to obtain our conclusion.
\end{proof}

This all leads up to a kind of anti-concentration result first described in \cite{hlz} which allowed them
(and will allow us) to substantially improve over what a naive union bound argument would achieve.
Whereas \cite{hlz} proved this for log-Lipschitz probability distributions, we obtain a similar result for the Gaussian distribution.
It will be the primary tool used to establish that the non-zero slack values are bounded away from $0$.
\begin{lemma}[Conditional anti-concentration]\label{lem:anti-concentration}
    Suppose $s_1,\dots,s_k \in \R$ are independently Gaussian distributed, each with standard deviation $\sigma > 0$,
    and suppose $t_1,\dots,t_k \in \R$ are fixed.
	Assume $q \in (0,1/e)$ is such that $\Pr[ \forall j \in [k] \colon  s_j \leq t_j] \geq q$.
	Then for any $\eps > 0$ we have
    \[
        \Pr\big[\exists j\in[k] : s_j \geq t_j-\eps \;\big|\; s \leq t \big]
        \leq q + 496 \eps \sigma^{-1} \ln^{3/2}(k/q).
    \]
\end{lemma}
\begin{proof}
	Since we are bounding a probability, without loss of generality we assume
	$\eps \leq \frac{\sigma}{496 \ln^{3/2}(k/q)}$.
    Define $C = \Big\{j \in [k] : \Pr\big[s_j \geq t_j - \eps \big] \geq q/k\Big\}$.
    We proceed by independence of the random variables to find
    \begin{align*}
        \Pr\big[ \exists j \in [k] \; : \; s_j \geq t_j - \eps \big| s \leq t \big]
        &\leq \sum_{j \in [k]} \Pr\big[ s_j \geq t_j - \eps \big| s_j \leq t_j \big]\\
        &\leq \sum_{j \in [k]\setminus C} \Pr\big[ s_j \geq t_j - \eps \big]
            + \sum_{j \in C}              \Pr\big[ s_j \geq t_j - \eps \big| s_j \leq t_j \big]\\
        &\leq q + \sum_{j \in C} \Pr\big[ s_j \geq t_j - \eps \big| s_j \leq t_j \big].
    \end{align*}
    For any $j \in C$ we know that $\Pr[s_j \geq t_j - \eps] \geq q/k$.
    The assumption of $\Pr[s \leq t] \geq q$ implies that $\Pr[s_j \leq t_j] \geq q \geq q/k$,
    and so we satisfy the conditions of \Cref{lem:gaussian-condition-reversing} and conclude
    \begin{align*}
        \sum_{j \in C} \Pr\big[s_j \geq t_j - \eps \big| s_j \leq t_j \big]
        &\leq \sum_{j \in C} 248 \eps \sigma^{-1} \sqrt{\ln(k/q)} \Pr[s_j \geq t_j]\\
        &= 248 \eps \sigma^{-1}\sqrt{\ln(k/q)} \cdot \E\Big[\abs*{\{j \in C \; : \; s_j \geq t_j\}}\Big].
    \end{align*}
    Denote this last random set as $V = \{j \in C \; : \; s_j \geq t_j\}$.
    Now recall the Chernoff bound (\Cref{chernoff}) which establishes that
    $q \leq \Pr[s \leq t] = \Pr\big[\abs{V} = 0\big] \leq \exp(-\E[\abs{V}]/2)$.
    Taking all of the above together we find
    \begin{align*}
        \Pr\big[ \exists j \in [k] \; : \; s_j \geq t_j - \eps \big| s \leq t \big]
        &\leq q + \sum_{j \in C} \Pr\big[ s_j \geq t_j - \eps \big| s_j \leq t_j \big] \\
        &\leq q + 248 \eps \cdot \sigma^{-1} \sqrt{\ln(k/q)} \cdot \E[\abs{V}] \\
        &\leq q + 496 \eps \sqrt{\ln (k/q)} \ln(1/q),
    \end{align*}
    finishing the proof.
\end{proof}

With these technical prerequisites in place, we can now prove the main result of this subsection.
Let us define the main properties of interest.
\begin{definition}\label{def:feasible-gap-bases}
    For a matrix $A \in \R^{n \times d}$ and vector $b \in \R^n$, define the set of feasible bases as
    \[
        F(A,b) = \{I \in \binom{[n]}{d} : \text{$A_I$ invertible and~} A x^I \leq b\}.
    \]
    Following that, define the set of \emph{feasible bases with relative gap} $g > 0$ as
    \[
        G(A,b,g) = \{I \in F(A,b) : A_{[n]\setminus I} x^I \leq b_{[n]\setminus I} - g\cdot\norm{x^I}\cdot \mathbf{1}\}.
    \]
\end{definition}

For an appropriate choice of $g$, we prove that the set $G(A,b,g)$ contains most of the set $F(A,b)$ on average.
\begin{theorem}(Slacks are large)\label{thm:slacks}
    Let the matrix $A \in \R^{n \times d}$ and index set $I \in \binom{[n]}{d}$ be as follows.
    We assume the entries of the submatrix $A_I$ to be fixed, with $A_I$ invertible,
    and we assume the remainder $A_{[n]\setminus I}$ to have independent Gaussian distributed entries,
    each with standard deviation $\sigma > 0$.
    Take $b \in \R^n$ to be fixed.
    If $\Pr[I \in F(A,b)] \geq 2 n^{-d}$ then
    \[
        0.9 \Pr[I \in F(A,b)] \leq
        \Pr\Big[I \in G(A,b,\frac{\sigma}{5000 d^{3/2} \ln(n)^{3/2}}) \Big] + n^{-d}.
    \]
\end{theorem}
\begin{proof}
    Compute $x^I = A_I^{-1}b_I$.
    For each
    $j \in [n]\setminus I$ define $t_j = b_j/\norm{x^I}$ and $s_j = a_j^\T x^I /\norm{x^I}$,
    thus defining two vectors $s, t \in \R^{n-d}$.
    We find that
    $I \in F(A,b)$ is equivalent to the system of inequalities $s \leq t$.
    Observe that $t$ is fixed and that the entries of $s$ are independently Gaussian distributed, each with standard deviation $\sigma$.
    Observe that, by assumption, we have that $\Pr[s \leq t] = \Pr[A x^I \leq b] \geq 2n^{-d}$.

    Taking $\eps = \frac{\sigma}{5000 d^{3/2} \ln(n)^{3/2}}$, observe that $I \in G(A,b,\eps)$ is equivalent to the system of inequalities $s \leq t-\eps$.
    Plugging in the conditional anti-concentration \Cref{lem:anti-concentration} with $q=n^{-d}$
	and $k=n-d$, using that $d \geq 3$ gives that
    \begin{align*}
	    \Pr[I \notin G(A,b,\eps) \mid I \in F(A,b)]
            &= \Pr[\exists j \in [n-d] : s_j > t_j - \eps \mid s \leq t] \\
            &\leq 496 \eps \sigma^{-1} \ln(n^{d+1})^{3/2} + n^{-d}\\
            &\leq 0.1 + n^{-d}.
        \end{align*}
    Equivalently, this gives $\Pr[I \in G(A,b,\eps) \mid I \in F(A,b)] \geq 0.9 - n^{-d}$.
    Multiplying both sides by $\Pr[I \in F(A,b)]$ and remembering that $G(A,b,\eps) \subseteq F(A,b)$ gives
    \[
        \Pr[I \in G(A,b,\eps)] \geq (0.9 - n^{-d})\Pr[I \in F(A,b)] \geq 0.9 \Pr[I \in F(A,b)] - n^{-d}.
    \]
    This is equivalent to the desired conclusion.
\end{proof}

\subsection{Triples}\label{sec:triples}

In order to get an upper bound on the size of the shadow path, we want to reason about
the case where the shadow path contains a basis in $M(A,c+Z,c'+Z,m)$ and whose neighbors
on the shadow path are in $G(A,b,g)$. In order to do this effectively without having to worry
about non-trivial correlations, we will consider sequences of three consecutive bases
on the shadow path which are all contained in $M(A,c+Z,c'+Z,m) \cap G(A,b,g)$.
When a large enough fraction of bases are in this intersection, then the number of triples of consecutive bases will automatically be large as well.

\begin{definition}\label{def:triples}
    For a graph $G = (V,E)$ and $S \subseteq V$, write $T^S \subseteq G$ for the vertices $v \in S$
    that have at least $2$ neighbors in $S$.
\end{definition}

\begin{lemma}\label{lem:triples}
    Consider a fixed finite set $U$ of possible elements.
    Let $S \subseteq V \subseteq U$ be two random sets such that $\Pr[I \in S \mid I \in V] \geq p > 2/3$ for every $I \in U.$
    Suppose that $G = (V,E)$ is a graph on vertex set $V,$ consisting of $k$ connected components,
    each of which is a cycle or a path.
    Then we have
    \[
        \E[\abs{V}] \leq \frac{2\E[k] + \E[\abs{T^S}]}{3p-2}.
    \]
\end{lemma}
\begin{proof}
Fix $S, V$ and $G$ for now.
    We denote the number of edges adjacent to a vertex $I$ as $\delta_G(I)$.
    We count the sum of the degrees $\delta_G(I)$ of vertices $I\in S$.
    Counting per vertex, we find a total of at least $2\abs{S}-2k$ since all vertices have degree at least $2$,
    except possibly the endpoints of connected components that are paths, i.e.,
    \[
        2\abs{S} - 2k \leq \sum_{I \in S} \delta_G(I).
    \]
    This sum counts every edge in the induced subgraph $G[S]$ twice. Every edge outside $G[S]$
    counted in the sum $\sum_{I \in S} \delta_G(I)$
    connects to a vertex in $V \setminus S$, contributing 1 to that vertex's degree.
    This implies that
    \[
        \sum_{I \in S} \delta_G(I) \leq
        2\abs{E(G[S])} + \sum_{I \in V \setminus S} \delta_G(I)
        \leq
        2\abs{E(G[S])} + 2\abs{V\setminus S}.
    \]
    To further upper bound this last quantity, observe that every edge in the subgraph $G[S]$
	connects to two vertices in $S$. Every vertex $I \in T^S$ has degree $2$ in $G[S]$,
	while every vertex $I \in S \setminus T^S$ has degree 0 or 1 in $G[S]$.
	From this we count the sum of the degrees in $G[S]$ and find that
        $2 \abs{E(G[S])} = \sum_{I \in S} \delta_{G[S]}(I) \leq 2\abs{T^S} + \abs{S\setminus T^S} = \abs{S} + \abs{T^S}$.
    Taking all of the above together, we find
    \[
	    2\abs{S} - 2k \leq \abs{S} + \abs{T^S} + 2\abs{V \setminus S} = \abs{T^S} + 2\abs{V} - \abs{S}.
    \]
    Simplifying, we get $3\abs{S} \leq 2k + \abs{T^S} + 2\abs{V}$.
    Now it is time to remember that everything is random to conclude
    \[
        3p\E[\abs{V}] \leq 3\E[\abs{S}] \leq 2\E[k] + \E[\abs{T^S}] + 2\E[\abs{V}],
    \]
    and hence we may rearrange to $\E[\abs{V}] \leq \frac{2\E[k]+\E[\abs{T^S}]}{3p - 2}$.
\end{proof}

Regard the shadow path as a random graph on the nodes $P(A,b,c+Z,c'+Z) \subseteq \binom{[n]}{d}$, with an edge between two bases on the path if and only if they are adjacent in the shadow path.
This subsection and the previous two all lead up to the following structural result:
\begin{theorem}\label{thm:bound-path-by-separated}
    Let the matrix $A \in \R^{n \times d}$ have independent Gaussian distributed entries, each with standard deviation $\sigma \leq \frac{1}{4\sqrt{d\ln n}}$. Assume furthermore that the rows of $\E[A]$ each have Euclidean norm at most $1$.
    Take $b \in \R^n$ and $c,c' \in \R^d$ to be fixed.
    If $Z \in \R^d$ has a $1$-log-Lipschitz probability density function then
    \[
        \E[\abs{P(A,b,c+Z,c'+Z)}] \leq 504 + 2\E[\abs{T^{G(A,b,\frac{\sigma}{5000d^{3/2} \ln(n)^{3/2}}) \cap M(A,c+Z,c'+Z, \frac{\ln(1/0.99)}{2d})}}].
    \]
\end{theorem}
\begin{proof}
    Write $m = \ln(1/0.99)/2d$ and $g=\frac{\sigma}{5000 d^{3/2} \ln(n)^{3/2}}$.
    For any $t \geq 0$, abbreviate $M(A,c+Z,c'+Z,t) = M(t)$ and $G(A,b,t) = G(t)$.
    Moreover, abbreviate $\mathcal{P} = P(A,b,c+Z,c'+Z)$.
    Recall that $I \in \mathcal{P}$ is equivalent to $I \in M(0) \cap G(0)$ and that $I \in G(0)$ is equivalent to $I \in F(A,b)$.

    Define $U := \{I \in \binom{[n]}{d} : \Pr[I \in \mathcal{P}] \geq 100 n^{-d}\}$.
    Note that $U$ is fixed, not random.
    Immediately we find that
    \begin{align*}
        \E[\abs{\mathcal{P}}] &\leq \E[\abs{\mathcal{P} \setminus U}] + \E[\abs{\mathcal{P} \cap U}] \\
                    &\leq 100 + \E[\abs{\mathcal{P} \cap U}].
    \end{align*}
    In order to use \Cref{lem:triples}, we will consider the universe $U$.
    Since the shadow path $\mathcal{P}$ is either a cycle or a path, we find that the graph $\mathcal{P} \cap U$
    has at most $k \leq \abs{\mathcal{P}\setminus U} + 1$ connected components, each of which is a cycle or path.

    For every basis $I \in U$, we will prove
    \[
        \Pr[I \in M(m) \cap G(g) \mid I \in \mathcal{P}] \geq 5/6.
    \]
    This is equivalent to the assertion that
    \[
        \Pr[I \in M(m) \cap G(g)] \geq \frac{5}{6}\Pr[I \in M(0) \cap G(0)],
    \]
    which we will prove.
    Pick $I \in U$ arbitrarily. We denote by $A_I$ the submatrix of $A$
    containing only the rows indexed by $I$, and we denote by $A_{[n]\setminus I}$ the remainder of $A$.
    For any fixed $t \geq 0$, the event $I \in M(t)$ depends on $A_I$ and $Z$
    but is independent of $A_{[n]\setminus I}$.
    Moreover, the event $I \in G(t)$ depends on $A_I$ and $A_{[n]\setminus I}$ but is independent of $Z$.
   We can change the order of integration and find that
    \begin{align}
        \Pr_{A,Z}\Big[I \in G(0) \cap M(0) \Big]
            &= \E_{A_I}\Big[ \E_{A_{[n]\setminus I}}\Big[\E_Z\Big[1[I \in G(0)] \cdot 1[ I \in M(0)]\Big]\Big]\Big] \nonumber\\
            &= \E_{A_I}\Big[ \E_{A_{[n]\setminus I}}\big[1[I\in G(0)]\big] \cdot \E_Z\big[1[ I \in M(0)]\big]\Big]\nonumber\\
            &= \E_{A_I}\Big[ \Pr_{A_{[n]\setminus I}}[I \in G(0)] \cdot \Pr_Z[ I \in M(0)]\Big]. \label{eq:first}
    \end{align}
    Note that we write $1[D]$ to denote the indicator function of an event $D$, i.e., $1[D]=1$ when $D$ and $1[D] = 0$ when $\neg D$, hence
    $\Pr[D] = \E[1[D]]$.
    Let $E_I$ denote the event that every row of $A_I$ has Euclidean norm at most $1+4\sigma\sqrt{d \ln n} \leq 2$.
    By \Cref{lem:gaussian-tail} we know that $\Pr[\neg E_I] \leq n^{-4d}$.
    Fix $A_I$ to be any matrix, assuming only it is invertible.
    If $E_I$ holds then we may use \Cref{cor:multipliers} to find
    \[
        \Pr_Z[I \in M(0)] \leq \frac{1}{0.99} \Pr[I \in M(m)] = \frac{1}{0.99}\Pr[I \in M(m)] + 1[\neg E_I].
    \]
    If $E_I$ does not hold then we directly find
    \[
        \Pr_Z[I \in M(0)] \leq 1 = 1[\neg E_I] \leq \frac{1}{0.99}\Pr[I \in M(m)] + 1[\neg E_I].
    \]
    Thus we have an upper bound for $\Pr_Z[I \in M(0)]$ that always holds. Now plug this bound into \eqref{eq:first} to find
    \begin{align}
        \Pr_{A,Z}\Big[I \in G(0) \cap M(0) \Big]
            &= \E_{A_I}\Big[ \Pr_{A_{[n]\setminus I}}[I \in G(0)] \cdot \Pr_Z[ I \in M(0)]\Big] \nonumber\\
        &\leq \E_{A_I}\Big[ \Pr_{A_{[n]\setminus I}}[I \in G(0)] \cdot \Big(\frac{1}{0.99}\Pr_Z[ I \in M(m)] + 1[\neg E_I]\Big)\Big]  \nonumber\\
        &\leq \E_{A_I}\Big[ \frac{1}{0.99}\Pr_{A_{[n]\setminus I}}[I \in G(0)] \cdot \Pr_Z[ I \in M(m)] + 1[\neg E_I]\Big] \nonumber\\
        &\leq \frac{1}{0.99} \E_{A_I}\Big[ \Pr_{A_{[n]\setminus I}}[I \in G(0)] \cdot \Pr_Z[ I \in M(m)]\Big] + n^{-d}.\label{eq:second}
    \end{align}
    We now enter a new case distinction based on $A_I$.
    First assume that $\Pr_{A_{[n]\setminus I}}[I \in G(0)] > 2n^{-d}$.
    We know by \Cref{thm:slacks} that
    \[
        \Pr_{A_{[n]\setminus I}}[I \in G(0)] \leq \frac{1}{0.9}\Pr_{A_{[n]\setminus I}}[I \in G(g)] + \frac{1}{0.9} n^{-d}
        \leq \frac{1}{0.9}\Pr_{A_{[n]\setminus I}}[I \in G(g)] + 2n^{-d}.
    \]
    In the alternative case we immediately find
    $\Pr_{A_{[n]\setminus I}}[I \in G(0)] \leq 2n^{-d}$.
    Thus we have the same upper bound in both cases.
    Plugging this into \eqref{eq:second}, it follows that
    \begin{align}
        \Pr_{A,Z}\Big[I \in G(0) \cap M(0) \Big]
        &\leq \frac{1}{0.99} \E_{A_I}\Big[ \Pr_{A_{[n]\setminus I}}[I \in G(0)] \cdot \Pr_Z[ I \in M(m)]\Big] + n^{-d} \nonumber\\
        &\leq \frac{1}{0.99} \E_{A_I}\Big[ \left(\frac{1}{0.9}\Pr_{A_{[n]\setminus I}}[I \in G(g)] + 2n^{-d}\right) \cdot \Pr_Z[ I \in M(m)]\Big] + n^{-d}\nonumber \\
        &\leq \frac{1}{0.99\cdot 0.9} \E_{A_I}\Big[ \Pr_{A_{[n]\setminus I}}[I \in G(g)] \cdot \Pr_Z[ I \in M(m)]\Big] + 4n^{-d}.\label{eq:third}
    \end{align}
    Rearranging the order of integration once more, we have found that
    \[
        \Pr_{A,Z}\Big[I \in G(0) \cap M(0) \Big]
        \leq
        \frac{1}{0.99 \cdot 0.9} \Pr_{A,Z}\Big[I \in M(m) \cap G(g) \Big] + 4n^{-d}.
    \]
    Now recall that $I \in U$, meaning that $\Pr[I \in G(0) \cap M(0)] \geq 100n^{-d}$.
    Using this to upper bound the last term, we find $4n^{-d} \leq 0.04\Pr[I \in G(0)\cap M(0)]$.
    As such, we can calculate that $\Pr[I \in G(0)\cap M(0)] \leq \frac{6}{5} \Pr[I \in M(m) \cap G(g)]$.

    We can now use \Cref{lem:triples} with $\E[k] \leq \E[\abs{P(A,b,c+Z,c'+Z)\setminus U}] \leq 100$
    and $p = 5/6$ to get the result
    \begin{align*}
        \E[\abs{P}] &\leq \E[\abs{P\setminus U}] + \E[\abs{P \cap U}] \\
                    &\leq 100 + \frac{2\E[k] + \E[\abs{T^{M(m) \cap G(g) \cap U}}]}{3p-2} \\
                    &\leq 504 + 2\E[\abs{T^{M(m)\cap G(g) \cap U}}] \\
                    &\leq 504 + 2\E[\abs{T^{M(m) \cap G(g)}}].
    \end{align*}
\end{proof}

\subsection{Close and far neighbors}

In order to make use of \Cref{thm:bound-path-by-separated},
the remainder of this section is dedicated to giving an upper bound on
$T^{M(A,c,c',m) \cap G(A,b,g)}$.
From here on we think of the shadow path as lying on the boundary of the shadow polygon.

Any basis in $T^{M(A,c,c',m) \cap G(A,b,g)}$ will be accounted for in one of two ways,
depending on the distance to its neighbors as measured in the projection.
Recall the definition of neighbor from \Cref{def:neighbour}.
\begin{definition}\label{def:hermits}
    For a given matrix $A \in \R^{n \times d}$, a right-hand side $b \in \R^n$, a pair of objectives $c,c' \in \R^d$ and some threshold $0 < \rho \leq 1/2 $, we denote the set of shadow path elements at far distance from their neighbors by
    \[
        H(A,b,c,c',\rho) = \Big\{I \in P(A,b,c,c') : \;\forall I' \in N(A,b,c,c',I), \; \norm{\pi_{c,c'}(x^I) - \pi_{c,c'}(x^{I'})} \geq \rho\norm{\pi_{c,c'}(x^I)}\Big\}.
    \]
\end{definition}

These subsets of the shadow paths are an important part of the argument,
so we will first extend the conclusion of \Cref{lem:compose-paths} to these sets.
\begin{lemma}\label{lem:compose-hermits}
    Let $A \in \R^{n \times d}, b \in \R^n$ and let $c,c' \in \R^d$ be linearly independent objectives.
    Assume that $P(A,b,c,c')$ is a non-degenerate shadow path and $\rho \in (0,1/2]$.
    Then we have that $\abs{H(A,b,c,y,\rho) \cap H(A,b,y,c',\rho)} \leq 2$ for all $y \in [c,c']$.
    Moreover, if $y_1,y_2,\dots,y_k \in [c,c']$ are sorted in order, then $$\sum_{i=1}^{k-1} \abs{H(A,b,y_i, y_{i+1},\rho)} \leq \abs{H(A,b,c,c',\rho)} + 2k.$$
\end{lemma}
\begin{proof}
    Since $H(A,b,y,y',\rho) \subseteq P(A,b,y,y')$ for any $\rho$, the first statement is immediate from \Cref{lem:compose-paths}.
    For the second statement, we abbreviate $P(i) = P(A,b,y_i, y_{i+1})$ and $H(i) = H(A,b,y_i,y_{i+1},\rho)$
    and calculate
    \begin{align*}
        \sum_{i=1}^{k-1} \abs{H(i)}
        &= \sum_{i=1}^{k-1} \abs{P(i)} - \abs{P(i) \setminus H(i)}\\
        &\leq \abs{P(A,b,c,c')} + 2k - \sum_{i=1}^{k-1}\abs{P(i) \setminus H(i)}
    \end{align*}
    using \Cref{lem:compose-paths}. Next, consider any basis $I \in P(A,b,c,c')$.
    Either $I \in H(A,b,c,c',\rho)$, or there exists $i \in [k-1]$ such that $I \in P(i)\setminus H(i)$.
    This implies that
    \[\abs{P(A,b,c,c')} \leq \abs{H(A,b,c,c',\rho)} + \sum_{i=1}^{k-1} \abs{P(i)\setminus H(i)}.\]
    Rearranging, we find our intended conclusion of $\sum_{i=1}^{k-1} \abs{H(i)} \leq 2k + \abs{H(A,b,c,c',\rho)}$.
\end{proof}

The next lemma is the main reason for introducing the sets $H(A,b,c,c')$.
\begin{lemma}\label{lem:gap-multipliers-implies-separation}
    For any $A \in \R^{n \times d}$, any $b \in \R^n$, any $c,c' \in \R^d$ with a non-degenerate shadow path,
    and any $g, m > 0$,
    the set of consecutive triples with large multipliers and slacks satisfies, for $\rho \in (0,1/2]$,
    \[
        \abs{T^{G(A,b,g) \cap M(A,c,c',m)}} \leq \abs{H(A,b,c,c',\rho)} + \frac{\rho \cdot \angle(c,c') \cdot \max(\norm{c},\norm{c'})}{(1-\rho)\cdot gm} + 2.
    \]
\end{lemma}
\begin{proof}
    Trivially, the total number of $I \in T^{G(A,b,g) \cap M(A,c,c',m)}$ that satisfy $I \in H(A,b,c,c',\rho)$
    is at most $\abs{H(A,b,c,c',\rho)}$. For that reason,
    we need only prove that the total number of $I \in T^{G(A,b,g) \cap M(A,c,c',m)}$ satisfying $I \notin H(A,b,c,c',\rho)$
    is at most $\frac{\rho \cdot \angle(c,c') \cdot \max(\norm{c},\norm{c'})}{(1-\rho)\cdot gm}$.

    Let $I \in T^{G(A,b,g) \cap M(A,c,c',m)} \setminus H(A,b,c,c',\rho)$. This implies that its neighbors $J, J' \in N(A,b,c,c',I)$ satisfy that $J,J' \in G(A,b,g)$. Furthermore, we know that $I \in M(A,c,c',m)$.
    Since $I \in M(A,c,c',m)$ we can take $y \in [c,c']$ such that $y^\T A_I^{-1} \geq m \mathbf{1}$.
    Write $i$ for the unique element $I \setminus J = \{i\}$. We directly compute
    \begin{align*}
        y^{\T}(x^I - x^J)
        &= (y^{\T}A_{I}^{-1})(A_{I} x^{I} - A_I x^J) \\
        &= (y^{\T}A_{I}^{-1})_i(A_{I} x^{I} - A_I x^J)_i \\
        &\geq  m\cdot (b_{I} -(A_{I} x^J))_i \\
        &\geq  m g \norm{x^J} \\
        &\geq  m g \norm{\pi_{c,c'}(x^J)},
    \end{align*}
    and similarly $y^\T(x^I-x^{J'}) \geq mg\norm{\pi_{c,c'}(x^{J'})}$.
    Since $I \notin H(A,b,c,c',\rho)$ we must have at least one close-by neighbor. Without loss of generality
    assume $J$ satisfies $\norm{\pi_{c,c'}(x^I - x^J)} \leq \rho \norm{\pi_{c,c'}(x^I)}$.
    For $J$ we can now express that $x^I$ and $x^J$ are far apart as measured by the inner product with $y$,
    for we must have
    \[
        \norm{\pi_{c,c'}(x^J)} \geq \norm{\pi_{c,c'}(x^I)} - \norm{\pi_{c,c'}(x^I - x^J)} \geq (1-\rho)\norm{\pi_{c,c'}(x^I)},
    \]
    implying that $y^\T x^I \geq y^\T x^J + (1-\rho)\cdot mg \cdot \norm{\pi_{c,c'}(x^I)}$.

    \begin{figure}
        \centering
        \definecolor{qqwuqq}{rgb}{0,0.39215686274509803,0}
        \definecolor{uuuuuu}{rgb}{0.26666666666666666,0.26666666666666666,0.26666666666666666}
        \definecolor{ududff}{rgb}{0.30196078431372547,0.30196078431372547,1}
        \definecolor{xdxdff}{rgb}{0.49019607843137253,0.49019607843137253,1}
        \begin{tikzpicture}[line cap=round,line join=round,>=triangle 45,x=0.8cm,y=0.8cm]
        \clip(-0.620007729313798,-2.0911247374460724) rectangle (16.632395676757376,5.590331897182605);
        \draw [shift={(6,4)},line width=2pt,color=qqwuqq,fill=qqwuqq,fill opacity=0.10000000149011612] (0,0) -- (63.66980504598516:1) arc (63.66980504598516:123.6900675259798:1) -- cycle;
        \draw [shift={(0,0)},line width=2pt,color=qqwuqq,fill=qqwuqq,fill opacity=0.10000000149011612] (0,0) -- (0:1) arc (0:33.690067525979785:1) -- cycle;
        \draw [line width=2pt,dash pattern=on 5pt off 5pt,domain=-0.620007729313798:16.632395676757376] plot(\x,{(-0-0*\x)/20});
        \draw [line width=2pt] (0,0)-- (6,4);
        \draw [line width=2pt] (6,4)-- (15.973200724483275,-0.9356042273284247);
        \draw [line width=2pt,dotted] (6,36.40001951157273)-- (6,0);
            \draw[color=xdxdff] (0.5,-0.35) node {$\pi_{c,c'}(x^J)$};
            \draw[color=qqwuqq] (5.7,5.3) node {$\alpha_I$};
            \draw[color=ududff] (7,4.1) node {$\pi_{c,c'}(x^I)$};
            \draw[color=ududff] (6,-0.4) node {$Q$};
            \draw[color=ududff] (15,-1.2) node {$\pi_{c,c'}(x^{J'})$};
            \draw[color=black] (2.7,2.9) node {$\leq \rho\norm{\pi_{c,c'}(x^I)}$};
            \begin{small}
            \draw[color=black] (7.9,1.2) node {$\geq \frac{(1-\rho)mg\norm{\pi_{c,c'}(x^I)}}{\norm y}$};
            \end{small}
        \begin{scriptsize}
        \draw [fill=xdxdff] (0,0) circle (2.5pt);
        \draw [fill=ududff] (6,4) circle (2.5pt);
        \draw [fill=ududff] (15.973200724483275,-0.9356042273284247) circle (2.5pt);
        \draw [fill=xdxdff] (6,0) circle (2pt);
        \end{scriptsize}
        \end{tikzpicture}
        \caption{The plane $\linsp(c,c')$ in \Cref{lem:gap-multipliers-implies-separation}. The vector $y$ points straight up. The angles at $Q$ are both right angles.}
        \label{fig:shadow}
    \end{figure}

    In \Cref{fig:shadow}, we draw the points $\pi_{c,c'}(x^I), \pi_{c,c'}(x^J)$ and $\pi_{c,c'}(x^{J'})$.
    We find the line in $\linsp(c,c')$ which is orthogonal to $y$ and passes through $\pi_{c,c'}(x^J)$,
    and draw a single additional point $Q$, which is the orthogonal projection of $\pi_{c,c'}(x^I)$ onto said line.
   The triangle $\triangle(\pi_{c,c'}(x^J),Q,\pi_{c,c'}(x^I))$ has a right angle at $Q$.

    Let $\alpha_I$ denote the exterior angle of the shadow polygon at the vertex $\pi_{c,c'}(x^I)$, also drawn in the figure.
    By chasing angles we find that $\alpha_I \geq \angle(\pi_{c,c'}(x^I), \pi_{c,c'}(x^J), Q)$.
    This latter angle we can lower bound with its sine
    \[
        \angle(\pi_{c,c'}(x^I), \pi_{c,c'}(x^J), Q) \geq
        \sin(\angle(\pi_{c,c'}(x^I), \pi_{c,c'}(x^J), Q)) =
        \frac{\norm{\pi_{c,c'}(x^I) - Q}}{\norm{\pi_{c,c'}(x^J) - \pi_{c,c'}(x^I)}}
        \geq \frac{(1-\rho)mg}{\rho \norm{y}},
    \]
    using the lower bound on the length of the opposite edge and
    the upper bound on the length of the hypotenuse described in the text above,
    canceling the two factors of $\norm{\pi_{c,c'}(x^I)}$.
    What we have found is that for every $I \in T^{G(A,b,g) \cap M(A,c,c',m)} \setminus H(A,b,c,c',\rho)$
    we have $$\alpha_I \geq \frac{(1-\rho)mg}{\rho \norm{y}} \geq \frac{(1-\rho)mg}{\rho \max(\norm{c},\norm{c'})}.$$

    It is now that we note that, for all $I \in P(A,b,c,c')$ except the two endpoints,
    the exterior angles at the shadow vertices pack into the angle between the objectives $c,c'$.
    In particular, for the sum over $I$'s currently under consideration, except the endpoints, we have
    \[
        \sum_{\substack{I \in T^{G(A,b,g)\cap M(A,c,c',m)}\setminus H(A,b,c,c',\rho)\\ \text{not an endpoint}}} \frac{(1-\rho)mg}{\rho \max(\norm{c},\norm{c'})} \leq
        \sum_{\substack{I \in P(A,b,c,c')\\ \text{not an endpoint}}} \alpha_I \leq \angle(c,c').
    \]
    Observe that the first summand does not depend on $I$. Therefore we can divide through $\frac{(1-\rho)mg}{\rho \max(\norm{c},\norm{c'})} $ on all sides.
    Accounting for the possible contributions by the endpoints, we find
    \[
        \abs{T^{G(A,b,g)\cap M(A,c,c',M)} \setminus H(A,b,c,c',\rho)}
        \leq \frac{\rho \max(\norm{c},\norm{c'})\cdot\angle(c,c')}{(1-\rho) mg} + 2.
    \]
    This is exactly what was needed to finish the proof.
\end{proof}

Our next immediate concern is to bound the number of bases in $H(A,b,c,c',\rho)$.
Recall the case distinction which we sketched in \Cref{sec:overview}.
The overwhelming majority of bases in $H(A,b,c,c',\rho)$ will be treated in the following case, consisting of those bases $I \in H(A,b,c,c',\rho)$ for which $r \leq \norm{\pi(x_I)} \leq R$.
Those bases for which $\norm{\pi(x_I)}>R$ or $\norm{\pi(x_I)}<r$ will be treated separately in \Cref{sec:norms}.

In order to count the bases in $H(A,b,c,c',\rho)$ with $r \leq \norm{\pi(x_I)} \leq R$, we charge them to an integral: every basis will ``cost'' some quantity of integral, and the total amount of integral available is limited.
Hence the number of bases is bounded from above.

\begin{definition}
    We define the ring with inner radius $r$ and outer radius $R$ as $D(R,r) = R \ball^2 \setminus r\ball^2$.
\end{definition}

\begin{lemma}\label{lem:donut}
    Let $T \subseteq \R^2$ be a closed convex set, and let $R > r > 0$.
    Then we can upper bound the following integral as follows
    \[
        \int_{D(R, r) \cap \partial T} \|x\|^{-1} \dd x \leq 4 \pi \lceil \log_2(R/r) \rceil.
    \]
\end{lemma}
\begin{proof}
    To start, we define for $i=1,\dots,l=\lceil \log_2(R/r) \rceil$ the ring $D_i \coloneqq D(2^ir, 2^{i-1}r)$.
    Note that $\bigcup_{i=1}^l D_i \supseteq D(R,r)$.
    We break up the large integral into smaller parts as
    \begin{equation} \label{eq:donut-integral-smaller-parts}
        \int_{D(R,r) \cap \partial T} \|x\|^{-1} \dd x \leq \sum_{i=1}^l \int_{D_i \cap \partial T} \|x\|^{-1} \dd x.
    \end{equation}

    For each $i=1,\dots,l$ we know that $x \in D_i$ implies an upper bound on the integrand $\|x\|^{-1} \leq \frac{1}{r2^{i-1}}$.
    We will now upper bound the size of the integration domain. Take any point $x \in D_i \cap \partial T$.
    Since $x$ is on the boundary of the convex set $T$ there exists a nonzero vector $y \in \R^2$ such that
    $y^\T x = \max_{x' \in T} y^\T x'$. This same vector $y$ demonstrates that
    $y^\T x = \max_{x' \in D_i \cap \partial T} y^\T x'$,
    and hence our point $x$ is also on the boundary of the  set $\conv(D_i \cap \partial T)$.
    It follows that our integration domain satisfies the inclusion
    \[
        D_i \cap \partial T \subseteq \partial\conv(D_i \cap \partial T)
    \]
    and hence $\int_{D_i \cap \partial T} \dd x \leq \int_{\partial\conv(D_i \cap \partial T)} \dd x$.
    We know that $\conv(D_i \cap \partial T) \subseteq 2^i r \ball^2$ is convex. By the monotonicity of surface area
    for inclusions of convex sets we find that $\int_{\partial\conv(D_i \cap \partial T)} \dd x \leq \int_{\partial 2^i r \ball^2} \dd x \leq 2\pi \cdot 2^i r$.
    Taken together, we have found that
    \[
        \int_{D_i \cap \partial T} \|x\|^{-1} \dd x
        \leq \int_{D_i \cap \partial T} \frac{1}{r2^{i-1}} \dd x
        \leq 4\pi.
    \]
    Summing over all values of $i=1,\dots,\lceil \log_2(R/r) \rceil$ in \eqref{eq:donut-integral-smaller-parts}
    gives the result.
\end{proof}

With this we find an upper bound on the size of $H(A,b,c,c',\rho)$ that is independent of the objectives $c,c'$.
\begin{lemma}\label{lem:donut-argument}
    For given constraint data $A \in \R^{n\times d}, b \in  \R^n$, objectives $c,c'\in \R^d$ and threshold $\rho > 0$
    we have for any $R > r > 0$ that
    \[
        \abs{H(A,b,c,c',\rho)} \leq \frac{128 \log_2(R/r)+208}{\rho}
        + \abs{\{I \in F(A,b) : \norm{\pi_{c,c'}(x^I)} \notin [r,R]\}}.
    \]
\end{lemma}
\begin{proof}
    Abbreviate $S = \{I \in H(A,b,c,c',\rho) : \norm{\pi_{c,c'}(x^I)} \in [r,R]\}$.
    It is our goal to show that $\abs{S} \leq \frac{128 \log_2(R/r)+208}{\rho}$.

    For every $I \in P(A,b,c,c')$ and every $I' \in N(A,b,c,c',I)$ we abbreviate
    the norm $\ell_I = \norm{\pi_{c,c'}(x^I)}$ and the line segment
    $L_{I,I'} = [\pi_{c,c'}(x^I), \pi_{c,c'}(x^{I'})]$.

    Every such line segment is an edge of the shadow polygon
    $\pi_{c,c'}(\{x : Ax \leq b\})$. Notably, every edge is found at most twice in this manner: once as $L_{I,I'}$ and once as $L_{I',I}$.
    With that knowledge and using \Cref{lem:donut}, we can upper bound following the sum of integrals as
    \[
        \sum_{I \in S}
            \sum_{I' \in N(A,b,c,c',I)}
                \int_{D(2R, r/2) \cap L_{I,I'}}
                    \frac{1}{\norm{t}}
                \dd t
        \leq
        2\int_{D(2R,r/2) \cap \partial \pi_{c,c'}(\{x : Ax \leq b \})} \frac{1}{\norm{t}} \dd t \leq 8\pi\lceil\log_2(4R/r)\rceil.
    \]
    Since $R > r$ we know that $8 \pi \lceil \log_2(4R/r)\rceil \leq 32 \log_2(R/r) + 52$.

    Notice that if $I \in S$ then the intersection
    $L_{I,I'}\cap D(2\ell_I,\ell_I/2)$ contains a line segment of length
    at least $\rho \ell_I/2$, and on this line segment the integrand is at least $\frac{1}{2\ell_I}$.
    This implies that the integral on that line segment is at least
    $\int_{D(2R,r/2) \cap L_{I,I'}} \norm{t}^{-1} \dd t \geq \int_{D(2\ell_I,\ell_I/2) \cap L_{I,I'}} \norm{t}^{-1} \dd t \geq \rho/4$
    and we can lower bound the sum of integrals as
    \[
        \sum_{I \in S}
            \sum_{I' \in N(A,b,c,c',I)}
                \int_{D(2R, r/2) \cap L_{I,I'}}
                    \frac{1}{\norm{t}}
                \dd t
        \geq
        \sum_{I \in S}
            \sum_{I' \in N(A,b,c,c',I)}
            \frac{\rho}{4}
        \geq
        \frac{\rho \abs{S}}{4}.
    \]
    We thus learn that $\abs{S} \leq \frac{4 \cdot (32 \log_2(R/r)+52)}{\rho} = \frac{128 \log_2(R/r)+208}{\rho}$.
    Since $H(A,b,c,c',\rho) \subset S \cup \{I \in F(A,b):\norm{\pi_{c,c'}(x^I)} \notin [r,R]\}$,
    this finishes the proof.
\end{proof}

\subsection{Summing over subpaths}

We are almost ready to prove our key theorem for upper bounding the smoothed shadow size.
The last remaining issue that needs resolving is that \Cref{lem:gap-multipliers-implies-separation}
depends linearly on the norm of the longest objective vector.
The next lemma will help offset that growing factor with a proportionally shrinking angle.

\begin{lemma}\label{lem:bound-angle}
    Let $c,z \in \R^d$ with $ \norm{c}=1$.
    If $i \geq \lceil \log_2(\frac{\norm{z}+d}{\norm{c}})\rceil + 2$ then $\angle(2^{i-1}c + z, 2^i c + z) \leq \frac{5(\norm{z}+d)}{2^{i+1}\norm{c}}$.
\end{lemma}
\begin{proof}
    Note that $\angle(2^{i-1}c + z, 2^i c + z) = \angle(2^i c + 2z, 2^i c + z)$ and
    consider the triangle $\triangle(0, 2^i c + 2 z, 2^i c + z)$.
    Let us abbreviate its vertices by $a = 2^i c + 2z$ and $b = 2^i c + z$, so that our triangle is $\triangle(0,a,b)$.
	The assumption on $i$ gives that $\norm{z}+d \leq \norm{2^{i-1} c}/2$.
    By the triangle inequality we find that the edge $[a,b]$
    has the shortest length of the three since $\norm{a-b}=\norm{z} \leq \norm{2^i c} - 2\norm{z} \leq \min(\norm{2^i c + 2z},\norm{2^i c + z})$.

    We recall the law of sines to derive
    \[
        \frac{\sin(\angle(a,b))}{\norm{a-b}}
        \leq \frac{\sin(\angle(0-a, b-a)}{\norm{b}}
        \leq \frac{1}{\norm{b}} \leq \frac{1}{2^{i-1}\norm{c}}.
    \]
	This gives an upper bound $\sin(\angle(a,b)) \leq \frac{\norm{z}+d}{2^{i-1}\norm{c}}$ on the sine of our desired angle.
	To relate this to the angle itself,
    note that the shortest edge of a triangle is opposite of the smallest angle, which gives us
    that $\angle(a, b) \leq \pi/3$.
    For any angle $\alpha \in [0, \pi/3]$ one has $\sin(\alpha) > 0.8 \alpha$,
    so in particular $\angle(2^{i-1}c + z, 2^i c +z) \leq \frac{5(\norm{z}+d)}{2^{i+1}\norm{c}}.$
\end{proof}

We are now able to break up a long shadow path into smaller subpaths for easier analysis.
In both the above and the below lemma, there is the seemingly mysterious quantity $\norm{z}+d$.
For any readers who wish to make sense this summation in terms of dimensional analysis,
we note that $z$ will be sampled from a $1$-log-Lipschitz probability distribution, which makes $\E[\norm{z}] \geq d$.
This ensures $\norm{Z}$ and $d$ are commensurable.

\begin{lemma}\label{lem:long-shadow-path}
    Consider constraint data $A \in \R^{n \times d}, b \in \R^n,$ and linearly independent objectives $z\in\R^d$, $c \in \sfe$
    with a non-degenerate shadow.
    Consider any values $g,m > 0$ and analysis parameters $R > r > 0$ and $k \in \mathbb{N}$.
    Assuming $(10 \log_2(R/r) +16) gm \leq d$, we find
    {\small
    \begin{align*}
        &\abs{T^{G(A,b,g) \cap M(A,z,2^k c+z,m)}} \\
    &\leq 2\sqrt{\frac{(128 \log_2(R/r) + 208)(5k+58)(\norm{z}+d)}{gm}  }
        + \abs{\{I \in F(A,b) : \norm{\pi_{c,Z}(x^I)} \notin [r,R]\}} + 4k+4.
   \end{align*}
    }
\end{lemma}
\begin{proof}
    For $y, y' \in \R^d$ we abbreviate $S(y,y') = T^{G(A,b,g) \cap M(A,y,y',m)}$
    and we cover the set of triples by smaller segments as follows
    \[
        S(z,2^k c+z) = S(z,c+z) \cup \bigcup_{i=1}^k S(2^{i-1} c + z, 2^i c + z).
    \]
    For the sake of succinctness write $c^{(-1)} = z$ and for each $i = 0,\dots,k$ write $c^{(i)} = 2^i c + z$.
    With this notation, our task is to bound
    \(
        \sum_{i=0}^k \abs{S(c^{(i-1)},c^{(i)})}.
    \)
    For each $i=0,\dots,k$ we apply \Cref{lem:gap-multipliers-implies-separation} and find, for some $\rho \in (0,1/2]$
    to be decided later, that
    {\small
    \begin{align}
        \sum_{i=0}^k \abs{S(c^{(i-1)},c^{(i)})}
        &\leq 2(k+1) + \sum_{i=0}^k \abs{H(A,b,c^{(i-1)},c^{(i)},\rho)} + \frac{\rho \cdot \angle(c^{(i-1)},c^{(i)}) \cdot\max(\norm{c^{(i-1)}},\norm{c^{(i)}})}{(1-\rho)\cdot gm} \nonumber \\
        &\leq 2(k+1) + \sum_{i=0}^k \abs{H(A,b,c^{(i-1)},c^{(i)},\rho)} + \frac{2 \rho \cdot \angle(c^{(i-1)},c^{(i)}) \cdot (\norm{z} + 2^i \norm{c})}{gm}. \label{eq:pathwise}
    \end{align}
    }%
    The sets $H(A,b,c^{(i-1)},c^{(i)},\rho)$ have pairwise overlap of at most $2$ by \Cref{lem:compose-hermits}, which gives with \Cref{lem:donut-argument} that
    \begin{align}
        \sum_{i=0}^k \abs{H(A,b,c^{(i-1)},c^{(i)},\rho)}
        &\leq \abs{H(A,b,c^{(-1)},c^{(k)},\rho)} + 2(k+1) \nonumber\\
        &\leq \frac{128 \log_2(R/r)+208}{\rho} + \abs{\{I \in F(A,b) : \norm{\pi_{z,c}(x^I)} \notin [r,R]\}} + 2(k+1). \label{eq:hermits}
    \end{align}
    Now for the angle terms in \eqref{eq:pathwise},
    we wish to bound $\sum_{i=0}^k \angle(c^{(i-1)},c^{(i)}) \cdot (\norm{z} + 2^i \norm{c}) $ from above.
    Observe that by construction the angles sum as $\sum_{i=0}^k \norm{z}\cdot \angle(c^{(i-1)},c^{(i)}) = \norm{z}\cdot\angle(c^{(-1)}, c^{(k)})$.
    We now subdivide into three parts based on a threshold $h = \lceil \log_2(\frac{\norm{z}+d}{\norm{c}}) \rceil + 2$ as
    {\small
    \begin{align}
        \sum_{i=0}^k \angle(c^{(i-1)},c^{(i)}) \cdot (\norm{z} + 2^i \norm{c})
        &=
        \angle(c^{(-1)},c^{(k)}) \cdot \norm{z} + \sum_{i=0}^k 2^i \angle(c^{(i-1)},c^{(i)}) \cdot \norm{c}\nonumber\\
        &\leq
        \pi \cdot \norm{z} +
        \sum_{i=0}^{h} 2^i \angle(c^{(i-1)},c^{(i)}) \cdot \norm{c}
        +
        \sum_{i=h+1}^k 2^i \angle(c^{(i-1)},c^{(i)}) \cdot \norm{c}. \label{eq:subdivide}
    \end{align}
    }%
    For the middle term, we have $i \leq h$ which implies $2^i \norm{c} \leq 8(\norm{z}+d)$
    and hence the partial sum satisfies $\sum_{i=0}^h 2^i \angle(c^{(i-1)},c^{(i)}) \cdot \norm{c} \leq 8\pi (\norm{z}+d)$.
    For the third term we use \Cref{lem:bound-angle} to find
	$\sum_{i=h+1}^k 2^i \angle(c^{(i-1)},c^{(i)})\cdot\norm{c} \leq \frac{5}{2}k(\norm{z}+d).$
    The three terms together thus sum up to
    \begin{align}
        \sum_{i=0}^k \angle(c^{(i-1)},c^{(i)}) \cdot (\norm{z} + 2^i \norm{c})
        &\leq \pi \norm{z} + 8\pi(\norm{z}+d) + \frac{5}{2}k(\norm{z}+d) \nonumber \\
        &\leq (29 + \frac{5}{2}k) \cdot (\norm{z} + d). \label{eq:angles}
    \end{align}
    Taking \eqref{eq:hermits}, \eqref{eq:subdivide} and \eqref{eq:angles} together we can now upper bound \eqref{eq:pathwise} as
    {\small
    \begin{align*}
        S(z,2^k c + z)
        &\leq 4k+4
                        + \frac{128 \log_2(R/r)+208}{\rho}
                        + \abs{\{I \in F(A,b) : \norm{\pi_{z,c}(x^I)} \notin [r,R]\}}
                        + \frac{\rho (58 + 5k) \cdot (\norm{z}+d)}{gm}.
      \end{align*}
  }
    It remains to choose $\rho \in (0,1/2]$ so as to find the strongest upper bound,
    which is attained at $\rho = \sqrt{\frac{(128 \log_2(R/r)+208) \cdot gm}{(58+5k) \cdot (\norm{z}+d)}}$.
    The assumption $(10 \log_2(R/r) +16) gm \leq d$ ensures that this choice satisfies $\rho \leq 1/2$.
\end{proof}

We are now ready to state our smoothed complexity technical theorem,
which we will apply with two different choices for $b$.
\begin{theorem}\label{thm:main}
    Let the constraint matrix $A \in \R^{n \times d}$ have independent Gaussian distributed entries,
    each with standard deviation $\sigma \leq \frac{1}{4\sqrt{d \ln n}}$,
    and such that the rows of $\E[A]$ each have norm at most $1$.
    Let $b \in \R^n, c \in \R^d$ be arbitrary and fixed, as well as the analysis parameters $R > r > 0$ and $t > 1$.
    Assume $\ln(R/r) + 1.6 \leq 250 d^3 \ln(n)$.
    If $Z \in \R^d$ has a $1$-log-Lipschitz probability density function that satisfies $\Pr[\norm{Z} \geq t] \leq n^{-d}$
    and is independent of $A$ then
    \[
        \E[\abs{P(A,b,Z,c)}] \leq O\left(\sqrt{\frac{\E[\norm{Z}+d] \log(R/r)}{\sigma} \sqrt{d^{7} \log^{5}(nt)}}\right)
        + 2\E\left[\abs{\{I \in F(A,b) : \pi_{c,Z}(x^I) \notin [r,R]\}}\right].
    \]
\end{theorem}
\begin{proof}
    We may assume $\sigma > n^{-2d}$, for otherwise the upper bound
    exceeds $n^d$ and is trivially true. We may assume without loss of generality that $\norm{c}=1$.

    Abbreviating a number of expressions, we write
    $k = 5d \left\lceil \log_2(nt)\right\rceil$
    to split up the shadow path into shorter segments,
    take the bound on the multipliers as $m = \ln(1/0.99)/2d$
    and the bound on the relative slacks as $g = \frac{\sigma}{5000 d^{3/2} \ln(n)^{3/2}}$.

    We have $\E[\abs{P(A,b,Z,c)}] \leq \E[\abs{P(A,b,Z,2^k c + Z)}] + \E[\abs{P(A,b,2^k c + Z, c)}]$.
    The second term is at most $7$ by \Cref{cor:small-angle-pivot-count}.
    For the first term we apply \Cref{thm:bound-path-by-separated} to find
    \[
        \E[\abs{P(A,b,Z, 2^k c + Z)}] \leq 504 + 2\E[\abs{T^{G(A,b,g)\cap M(A,Z,2^k c +Z, m)}}].
    \]
    Through applying \Cref{lem:long-shadow-path}, noting that indeed $(10 \log_2(R/r) +16) gm \leq d$, we find that
    \begin{align*}
        \E&\left[\abs{T^{G(A,b,g) \cap M(A,Z,2^k c + Z, m)}}\right] \\ &\le
        \E\left[2\sqrt{\frac{(128 \log_2(R/r) + 208)(5k+58)(\norm{Z}+d)}{gm}  }\right]
        + \E[\abs{\{I \in F(A,b) : \norm{\pi_{c,Z}(x^I)} \notin [r,R]\}}] + 4k+4,
    \end{align*}
    and Jensen's inequality gives that $\E[\sqrt{\norm{Z}+d}] \leq \sqrt{\E[\norm{Z} + d]}$.
    Filling in all values and noting that $t > 1$ obtains the conclusion.
\end{proof}

\subsection{Norms}\label{sec:norms}

In this subsection we will look at basic solutions with ``very small'' and ``very large'' norms and show that they are very unlikely to occur.
This will give us values for our parameter $R$ for ``very large'' norms and our parameter $r$ for ``very small'' norms.
For the former we will deduce guarantees in \Cref{lem:no-large-norms} and for the latter in \Cref{lem:no-small-norms} and \Cref{lem:no-small-projected-norms}.

\begin{lemma}\label{lem:no-large-norms}
    Let each row $a_i$, $i \in [n]$ of $A$ be an independent $\sigma^2$-Gaussian random variable,
    and let $b \in \R^n$ be fixed.
    Then we have for any $R > 0$ that
    \[
        \Pr\left[\max_{I \in \binom{[n]}{d}} \norm{ x^I} \geq R \norm{b}_\infty\right] \leq \frac{2 \cdot d^2 n^d}{\sigma R \sqrt{2\pi}}.
    \]
\end{lemma}
\begin{proof}
    Let $I \subseteq \binom{[n]}{d}$ denote the index set of a subset of rows of cardinality $d$.
    Let $ E_I$ denote the event that $\dist\big(a_j, \linsp(a_i : i \in I\setminus\{j\})\big) \geq d/R$
    holds for each $j \in I$.
    Note that if the matrix $A_I$ is invertible then the column of $A_I^{-1}$ corresponding to index $j \in I$
    has norm exactly equal to
    $1/\dist\big(a_j, \linsp(a_i : i \in I\setminus\{j\})\big)$.
    It follows by the triangle inequality that $ E_I$ implies that $x^I$ has norm at most
    $\norm{x^I} \leq \sum_{i \in I} \norm{(A_I^{-1})_i}\cdot \abs{b_i} \leq R \norm{b_I}_\infty$.
    Using this implication along with a union bound we find
    \begin{equation}\label{eq:large-union}
        \Pr\left[\max_{I \in \binom{[n]}{d}} \norm{ x^I} \geq R \norm{b}_\infty\right]
        \leq \Pr\left[\bigvee_{I \in \binom{[n]}{d}} \neg  E_I\right]
        \leq \sum_{I \in \binom{[n]}{d}} \Pr[\neg  E_I].
    \end{equation}

    It remains to show that $\Pr[\neg  E_I] \leq \frac{2d^2}{\sigma R \sqrt{2\pi}}$
    for all $I \in \binom{[n]}{d}$. Using another union bound, it suffices if we show
    for each $j \in I$ that
    \[
        \Pr\Big[\dist\big(a_j, \linsp(a_i : i \in I\setminus\{j\})\big) \leq d/R\Big] \leq \frac{2d}{\sigma R \sqrt{2\pi}}.
    \]
    We take the linear subspace $V = \linsp(a_i : i \in I\setminus\{j\})$ to be fixed,
    and take $y \in V^\perp \cap \sfe$ to be any fixed unit normal vector.
    Using this notation we can write
    \[
        \dist\big(a_j, V\big) = \abs{y^\T a_j}.
    \]
    Note that $V$ is defined only using $a_i$ with $i \in I \setminus\{j\}$,
    and in particular that $y$ is independent of $a_j$.
    That means that, after conditioning on the values of $a_i$ for $i\in I\setminus\{j\}$,
    the signed distance $y^\T a_j$ is Gaussian distributed with mean $y^\T \E[a_j]$
    and standard deviation $\sigma$.
    The distance can only be small if $y^\T a_j \in (-d/R, d/R)$ and hence we find
    \begin{align}
        \Pr\Big[\dist\big(a_j, V\big) \leq d/R\Big]
        &= \Pr\Big[\abs{y^\T a_j} \leq d/R\Big]\nonumber \\
        &= \Pr\big[y^\T a_j \in [-d/R,d/R]\big]\nonumber \\
        &\leq \frac{2d}{R} \cdot \frac{1}{\sigma\sqrt{2\pi}}, \label{eq:large-interval}
    \end{align}
    using the fact that the probability density function of $y^\T a_j$ is uniformly upper bounded by $1/\sigma\sqrt{2\pi}$.
    Combining \eqref{eq:large-union} with the union bound over all $j \in I$ and \eqref{eq:large-interval}
    proves the lemma.
\end{proof}

\begin{lemma}[No small norms]\label{lem:no-small-norms}
    Let the rows $a_i$, $i \in [n]$ of $A$ have independent Gaussian distributed entries
    with expectations of norm $\norm{\E[a_i]} \leq 1$ for $i=1,\dots,n$,
    each with standard deviation $\sigma \leq \frac{1}{4\sqrt{d\log n}}$.
    Let $b \in \R^n$ be arbitrary subject to $\abs{b_i} \geq \eps$ for all $i \in [n]$.
    Then we have
    \[
        \Pr\left[\min_{I \in \binom{[n]}{d}} \norm{x^I} < \eps/2\right] \leq n^{-d}.
    \]
\end{lemma}
\begin{proof}
    Assume that $\norm{a_i} \leq 2$ for all $i \in [n]$.
    Then for any $x \in \R^d$ with $\norm{x} < \eps/2$ it follows that
    $a_i^\T x \leq \norm{a_i} \cdot \norm{x} < \eps \leq b_i$.
    In particular this implies that $x$ cannot be obtained as $A_I^{-1}b_I$ for any
    $I \in \binom{[n]}{d}$ with $i \in I$.
    Thus if $\norm{a_i} \leq 2$ then any basic solution $x^I$ for $I \in \binom{[n]}{d}$
    must satisfy $\norm{x^I} \geq \eps/2$.
    This implication then gives
    $\Pr[\min_{I \in \binom{[n]}{d}} \norm{x^I} < \eps/2]
    \leq \Pr[\exists i \in [n] : \norm{a_i} > 2]$.

    For $a_1,\dots,a_n$ we note by the triangle inequality that
    $\norm{a_i} > 2$ implies $\norm{a_i - \E[a_i]} > 1$. We call on \Cref{cor:gaussian-globaldiam} to find that
    \begin{align*}
        \Pr[\exists i \in [n] : \norm{a_i} > 2]
        &\leq \Pr[\exists i \in [n] : \norm{a_i - \E[a_i]} > 1] \\
        &\leq \Pr\left[\exists i \in [n] : \norm{a_i - \E[a_i]} > 4\sigma\sqrt{d\log n}\right]
        \leq n^{-d}.
    \end{align*}
    We have thus found that $\Pr\left[\min_{I \in \binom{[n]}{d}} \norm{x^I} < \eps/2\right]
    \leq \Pr[\exists i \in [n] : \norm{a_i} > 2] \leq n^{-d}$ as required.
\end{proof}

\begin{lemma}\label{lem:no-small-projected-norms}
    Let the rows $a_i$, $i \in [n]$ of $A$ have independent Gaussian distributed entries
    with expectations of norm $\norm{\E[a_i]} \leq 1$ for $i=1,\dots,n$,
    and $\sigma \leq \frac{1}{4\sqrt{d \log n}}$.
    Let $c \in \R^d \setminus\{0\}$ and $b \in \R^n$ be fixed subject to
    $\abs{b_i} > \eps$ for every $i\in[n]$
    and let $Z \in \R^d$ be distributed independently from $A,b$
    and rotationally symmetric.
    Then we have, for $\alpha,\eps > 0$, that
    \[
        \Pr\Big[\min_{I \in \binom{[n]}{d}} \norm{\pi_{\linsp(c,Z)}(x^I)} < \frac{\alpha\cdot\eps}{2} \Big]
        \leq n^{-d} + \alpha n^d \cdot \sqrt{d e}.
    \]
\end{lemma}
\begin{proof}
    We start with a simple bound, writing
    \[
        \min_{I \in \binom{[n]}{d}}\norm{\pi_{\linsp(c,Z)}(x^I)}
        \geq
        \min_{I \in \binom{[n]}{d}}\norm{x^I} \cdot
        \min_{I' \in \binom{[n]}{d}}\frac{\norm{\pi_{\linsp(c,Z)}(x^{I'})}}{\norm{x^{I'}}}.
    \]
    Thus, if $\min_{I \in \binom{[n]}{d}}\norm{\pi_{\linsp(c,Z)}(x^I)} < \alpha\cdot\eps/2$ is small then that implies that at least one of
    $\min_{I \in \binom{[n]}{d}}\norm{x^I} < \eps/2$ or
    $\min_{I' \in \binom{[n]}{d}}\frac{\norm{\pi_{\linsp(c,Z)}(x^{I'})}}{\norm{x^{I'}}} < \alpha$
    holds.
    A union bound over these two events gives us
    \[
    \Pr\Big[\min_{I \in \binom{[n]}{d}} \norm{\pi_{\linsp(c,Z)}(x^I)} < \frac{\alpha\cdot\eps}{2} \Big]
        \leq
        \Pr[\min_{I \in \binom{[n]}{d}}\norm{x^I} \leq \eps/2] +
        \Pr[\min_{I' \in \binom{[n]}{d}}\frac{\norm{\pi_{\linsp(c,Z)}(x^{I'})}}{\norm{x^{I'}}} \leq \alpha].
    \]
    As we have proven in \Cref{lem:no-small-norms}, we have
    $\Pr[\min_{I \in \binom{[n]}{d}} \norm{x^I} \leq \eps/2] \leq n^{-d}$
    for the first summand.
    It remains to upper bound the second summand.
    For this, we start by observing that for any $I \in \binom{[n]}{d}$ we have
    $\frac{\norm{\pi_{\linsp(c,Z)}(x^{I})}}{\norm{x^{I}}} \geq \frac{\abs{Z^\T x^I}}{\norm{Z}\cdot\norm{x^I}}$.
    This inequality implies that if the former quantity is small then the second quantity must be small.
    This in turn results in the inequality
    \(
    \Pr\left[\min_{I \in \binom{[n]}{d}}\frac{\norm{\pi_{\linsp(c,Z)}(x^{I})}}{\norm{x^{I}}} \leq \alpha\right]
    \leq \Pr\left[\min_{I \in \binom{[n]}{d}} \frac{\abs{Z^\T x^I}}{\norm{Z} \cdot\norm{x^I}} \leq \alpha\right].
    \)
    To upper bound this last probability, we observe that for each $I \in \binom{[n]}{d}$ the fraction
    $\frac{Z^\T x^I}{\norm{Z}\cdot\norm{x^I}}$ has a distribution identical to the inner product $\theta^\T e_1$
    between a uniformly random unit vector $\theta \in \sfe$ and an arbitrarily chosen standard basis vector.
    Taking a union bound over all $\abs{\binom{[n]}{d}} \leq n^d$ choices of $I$, we bound
    \begin{align*}
        \Pr\left[\min_{I \in \binom{[n]}{d}} \frac{\abs{Z^\T x^I}}{\norm{Z} \cdot\norm{x^I}} \leq \alpha\right]
        &\leq
        \sum_{I \in \binom{[n]}{d}}
        \Pr\left[ \frac{\abs{Z^\T x^I}}{\norm{Z} \cdot\norm{x^I}} \leq \alpha\right] \\
        & \leq
        n^d \cdot \Pr\left[ \abs{\theta^\T e_1} \leq \alpha\right].
    \end{align*}
    Using \Cref{thm:sphere-mass} to upper bound $\Pr\big[\abs{\theta^\T e_1}\leq\alpha\big]\leq\alpha\sqrt{de}$ we obtain the result.
\end{proof}

\subsection{Conclusion}
We require the semi-random shadow bound for two cases, either when the entries of $b\in\R^n$
are all fixed to $1$, or when the entries of $b$ are Gaussian distributed. For the former we will present \Cref{thm:fixedrhs} and for the latter \Cref{thm:smoothedrhs}.

\begin{theorem}\label{thm:fixedrhs}
    Let the constraint matrix $A \in \R^{n \times d}$ have independent Gaussian distributed entries,
    each with standard deviation $\sigma > 0$
    and such that the rows of $\E[A]$ each have norm at most $1$. Let the right hand side vector $b$ be fixed to be $1$.
    Let $c \in \R^d$ be arbitrary and fixed, as well as the analysis parameters $R > 2r > 0$.
    If $Z \in \R^d$ has a $1$-log-Lipschitz probability density function that satisfies $\Pr[\norm{Z} \geq 2 e d \ln(n)] \leq n^{-d}$
    and is independent of $A$, then the semi-random shadow path on $\{x : Ax \leq 1\}$ has length bounded as
    \[
        \E\left[\abs{P(A,1,Z,c)}\right] \leq O\left(\sqrt{\frac{1}{\sigma} \sqrt{d^{11} \log^{7} n}} + d^3 \log(n)^2\right).
    \]
\end{theorem}
\begin{proof}
We distinguish three cases on the values of the standard deviation $\sigma.$
    If $\sigma < n^{-2d}$ the right-hand side exceeds $\binom{n}{d}$ and the result follows immediately.
    Thus we assume $\sigma \geq n^{-2d}$. For the second case, we further assume that $\sigma \leq \frac{1}{4\sqrt{d\log n}}$ such that we can apply \Cref{thm:main} in the following.
    Choose $R = n^{5d} \geq 2d^2 n^{2d}/\sigma\sqrt{2\pi}$ and $r= n^{-2d} \leq (2n^d \cdot \sqrt{de})^{-1}$.
    We apply \Cref{thm:main} and with these values, we obtain $\log(R/r) \leq O(d\log n)$.
    It remains to upper bound
    \[
        \E\left[\abs{\{I \in F(A,b) : \norm{\pi_{c,Z}(x^I)} > R\}}\right] + \E\left[\abs{\{I \in F(A,b) : \norm{\pi_{c,Z}(x^I)} < r\}}\right].
    \]

    Using \Cref{lem:no-large-norms} we get
    $\E\left[\abs{\{I \in F(A,b) : \norm{\pi_{c,Z}(x^I)} > R\}}\right] \leq
    n^d\Pr[\max_{I \in\binom{[n]}{d}} \norm{x^I} > R] \leq 1$.
    To bound the expected number of bases with small projected norms,
    we start similarly by
    \[
        \E[\abs{\{I \in F(A,b) : \norm{\pi_{c,Z}(x^I)} < r\}}] \leq
        n^d\Pr[\max_{I \in\binom{[n]}{d}} \norm{\pi_{c,Z}(x^I)} < r].
    \]
    We apply \Cref{lem:no-small-projected-norms} with $\eps=1$
    and $\alpha = (n^{2d} \cdot \sqrt{de})^{-1}$ to get
    $\Pr[\min_{I \in\binom{[n]}{d}} \norm{\pi_{c,Z}(x^I)} < r] \leq 2n^{-d}$,
    which finishes the argument since now
    $\E[\abs{\{I \in F(A,b) : \pi_{c,Z}(x^I) \notin [r,R]\}}] \leq 3$.

    For the third case we assume $\sigma > \frac{1}{4\sqrt{d\log n}}$ and calculate $\abs{P(A,1,Z,c)} \leq O(d^3 \log(n)^2)$ closing the proof.
\end{proof}

\begin{theorem}\label{thm:smoothedrhs}
    Let the constraint matrix $A \in \R^{n \times d}$ have independent Gaussian distributed entries,
    as well as the vector $b \in \R^n$,
    each with standard deviation $\sigma > 0$,
    and such that the rows of $\E[(A,b)]$ each have norm at most $1$.
    Let $c \in \R^d$ be arbitrary and fixed, as well as the analysis parameters $R > 2r > 0$.
    If $Z \in \R^d$ has a $1$-log-Lipschitz probability density function that satisfies $\Pr[\norm{Z} \geq 2 e d \ln(n)] \leq n^{-d}$
    and is independent of $A$, then the semi-random shadow path on $\{x : Ax \leq b\}$ has length bounded as
    \[
        \E\left[\abs{P(A,b,Z,c)}\right] \leq O\left(\sqrt{\frac{1}{\sigma} \sqrt{d^{11} \log^7 n}} + d^3 \log(n)^2 \right).
    \]
\end{theorem}
\begin{proof}
We distinguish three cases on the values of the standard deviation $\sigma.$
    If $\sigma < n^{-2d}$ the right-hand side exceeds $\binom{n}{d}$ and the result follows immediately.
    Thus we assume $\sigma \geq n^{-2d}$. For the second case, we further assume that $\sigma \leq \frac{1}{4\sqrt{d\log n}}$ such that we can apply \Cref{thm:main} in the following.
    From \Cref{cor:gaussian-globaldiam} it follows that with probability at most $n^{-d}$,
$\norm{b}_{\infty} > 1+ 4 \sigma \sqrt{d \log n}$.
Hence, we conclude that this scenario contributes at most $1$ shadow vertex to the expectation.
    Thus, we assume in the following that $\norm{b}_{\infty} \leq 1+ 4 \sigma \sqrt{d \log n}$.
    We choose $R= n^{5d} \geq \frac{2 d^2 n^{2d}}{\sigma}$
    and conclude applying \Cref{lem:no-large-norms} that
    as before,
    \[
        \E\left[\abs{\{I \in F(A,b) : \norm{\pi_{c,Z}(x^I)} > 2R\}}\right]
        \leq
        n^d\left(\Pr\left[\max_{I \in\binom{[n]}{d}} \norm{x^I} > \norm{b}_\infty R\right] + \Pr[\norm{b}_\infty > 2]\right) \leq 2.
    \]
    For this proof we pick $r= n^{-6d} \leq \frac{\sigma \sqrt{2 \pi}}{n^{3d}\sqrt{de}}$.
    We apply \Cref{thm:main} and with these values, getting $\log(2R/r) \leq O(d\log n)$.
    We are thus left with the task of bounding
    $\E[\abs{\{I \in F(A,b) : \norm{\pi_{c,Z}(x^I)} < r\}}]$.
    Take $\eps = \sigma \sqrt{2\pi} n^{-d}$
    and $\alpha = (n^{2d} \cdot \sqrt{de})^{-1}$.
    We separately treat the scenario where there exists $i \in [n]$ with $\abs{b_i} < \eps$
    and the scenario where for all $i \in [n]$ it holds that $\abs{b_i} \geq \eps$.

    In the first scenario we count at most $\binom{n}{d}$ bases $I$ with $\norm{\pi_{c,Z}(x^I)} < r$,
    and this scenario occurs with probability
    $\Pr[\exists i \in [n], \abs{b_i} < \eps] \leq \frac{2\eps}{\sigma\sqrt{2\pi}} < n^{-d}$.
    Thus this scenario contributes at most $1$ to the expectation.

    For the second scenario we apply \Cref{lem:no-small-projected-norms} to learn that
    $\Pr[\min_{I \in\binom{[n]}{d}} \norm{\pi_{c,Z}(x^I)} < r] \leq 2n^{-d}$,
    finding that this contributes at most $2$ to the expectation.
    This suffices for the theorem.
     For the third case we assume $\sigma > \frac{1}{4\sqrt{d\log n}}$ and calculate $\abs{P(A,1,Z,c)} \leq O(d^3 \log(n)^2)$ closing the proof.
\end{proof}

\section{Lower bound}\label{sec:lb}

In this section we will demonstrate that the exponent for $\sigma$ in the
shadow bound found in the previous section cannot be further improved without
significantly worsening the dependence on $n$.

\begin{definition}
    For $\eta > 0$ and $d \in \mathbb{N}$, a set $S \subset \sfe$ is called $\eta$-dense if for any $x \in \sfe$
    there exists $s \in S$ such that $\|x-s\| \leq \eta$.
\end{definition}

Dense sets have been previously studied, and in particular there are known bounds
on their size for greedy constructions.
\begin{lemma}[See, e.g., \cite{matousek} p.314]\label{lem:matousek}
    There exists an $\eta$-dense set $S \subset \sfe$ with cardinality $\abs{S} \leq (4/\eta)^d$.
\end{lemma}

For our unperturbed constraint data we will use a matrix whose rows form an $\eta$-dense set.
This will result in a feasible set which is ``close to the unit ball''.
\begin{lemma}\label{lem:lb-helper}
	Let $\{s_1,\dots,s_n\} \subset \sfe$ be $\eta$-dense, $\eta \leq 1/8$,
    and let $A \in \R^{n \times d}$ be a matrix with rows $a_1,\dots,a_n$.
    Assume that for every $i \in [n]$ we have $\norm{a_i - s_i} \leq \eta$.
    Given any vector $b \in [1-\eta, 1+\eta]^n$, the polyhedron $ \{x \in \R^d : Ax \leq b\} $ satisfies
    \[
	    (1-2\eta)\ball^d \subseteq \{x \in \R^d : Ax \leq b\} \subseteq (1+4\eta)\ball^d.
    \]
\end{lemma}
\begin{proof}
    Suppose $x \in \R^d$ satisfies $\norm{x}\leq 1-2\eta$.
    Consider any $i \in [n]$. By the triangle inequality we find
    $\norm{a_i} \leq \norm{s_i} + \norm{a_i - s_i} \leq 1+\eta$.
    By the Cauchy-Schwarz inequality we find
    \[
        a_i^\T x \leq \norm{a_i} \cdot \norm{x} \leq (1+\eta)(1-2\eta) = 1 - \eta - 2\eta^2 \leq b_i.
    \]
    Since this inequality $a_i^\T x \leq b_i$ holds for all $i \in [n]$ we conclude that any $x \in \R^d$ with $\norm{x} \leq 1-2\eta$
    satisfies $Ax \leq b$, implying that $(1-2\eta)\ball^d \subseteq \{x \in \R^d: Ax\leq b\}$.

    Now suppose $x \in \R^d$ satisfies $\norm{x} > 1+4\eta$.
	By the $\eta$-denseness of $S$ there exists an $i \in [n]$ such that
    $\norm{\frac{x}{\norm{x}} - s_i} \leq \eta$,
    and by assumption on $A$ we have
    $\norm{a_i - s_i} \leq \eta$.
    By the triangle inequality we know that $\norm{\frac{x}{\norm{x}} - a_i} \leq 2\eta$.
    We use the Cauchy-Schwarz inequality to find
    \begin{align*}
        a_i^\T x &= \norm{x} - \left(\frac{x}{\norm{x}} - a_i \right)^\T x \\
                 &\geq (1-2\eta)\norm{x} \\
                 &> (1-2\eta)(1+4\eta)  \\
                 &= 1 + 2\eta - 8\eta^2.
    \end{align*}
    Assuming that $\eta \leq 1/8$ gives us $a_i^\T x > 1+\eta \geq b_i$,
    implying $Ax \not\leq b$. Hence any $x \in \R^d$ for which $Ax \leq b$ must satisfy $\norm{x} \leq 1+4\eta$ and we find $\{x \in \R^d : Ax \leq b\} \subseteq (1+4\eta)\ball^d$.
\end{proof}

Finally we require one more lemma to bound the diameter, adapted from \cite{bdghl21}
\begin{lemma}
\label{lem:rel-diam}
	For $d \geq 2$, let $P \subseteq R\ball^d, R>0,$ be a simple bounded polytope
containing the origin in its interior and let
	\begin{align*}
		P^\circ &:= \{y \in \mathbb{R}^n: \sprod{x}{y} \leq 1, \forall x \in P\}
	\end{align*}
	denote the polar of $P$.
If every facet of $P^\circ$ has geometric diameter at most $\gamma > 0$,
	then for any unit-length objective vector $c \in \sfe$,
	any maximizing vertex $v^+ \in P$,
	and any minimizing vertex $v^- \in P$,
	there is no simplex path from $v^+$ to $v^-$
	of length less than $(d-1)(\frac{2}{R\gamma}-3)$.
\end{lemma}
\begin{proof}
	Since $P$ contains the origin in its interior, we may assume that we have a minimal inequality description
	$P = \{x \in \R^d : a_j^\T x \leq 1, \ j\in[n]\}$.
        Write $A \in \R^{n\times d}$ for the matrix with rows $a_1,\dots,a_n$.
        A standard result is that we have a minimal vertex description
	$P^\circ = \conv(a_1,\dots,a_n)$.
\noindent
	Any simplex path of length $k$ between $v^+$ and $v^-$ can be represented as
        a sequence $B^0,B^1,\dots,B^k \in \binom{[n]}{d}$ of feasible bases for $P$
	such that $\abs{B^{i-1} \cap B^{i}} = d-1$ for every $i=1,\dots,k$,
	and for which we have $a_j^\T v^+ = 1$ for all $j\in B^0$
	and $a_j^\T v^- = 1$ for all $j \in B^k$, as well as $c^\T A_{B^0}^{-1} \geq 0$ and $c^\T A_{B^k}^{-1} \leq 0$

	Construct a sequence of row indices of $A$ as follows:
	Write $\ell = \lfloor k/(d-1) \rfloor$.
	For every $t =1,\dots,\ell$ let $p^t \in B^{(d-1)t} \cap B^{(d-1)(t-1)}$ be arbitrary.
	Note that, due to the bases being adjacent $\abs{B^{i-1}\cap B^i}=d-1$ for all $i\in[k]$,
	all sets $B^{(d-1)t} \cap B^{(d-1)(t-1)}$ are non-empty, so our desired sequence of vectors exists.
	We will now lower bound $\ell$, which will then give a lower bound on the path length $k$.

        For every $i \in [k]$, we can write $v^i = A_{B^i}^{-1} 1$ for the vertex of $P$ corresponding
        to the feasible basis $B^i$.
        The set $F^i = \conv(a_j : j \in B^i)$ is a facet of the polar polytope $P^\circ$
        given by the facet-defining inequality $P^\circ \subseteq \{y : \sprod{y}{v^i} \leq 1\}$,
	hence $F^i$ has Euclidean diameter at most $\gamma$ by assumption.
	In particular, since $p^{t}, p^{t+1} \in B^{(d-1)(t+1)}$ for all $t=1,\dots,\ell$,
	we must have $a_{p^t},a_{p^{t+1}} \in F^{(d-1)(t+1)}$ and hence $\norm{a_{p^t} - a_{p^{t+1}}} \leq \gamma$.
        Since $B^0$ is a basis for $v^+$, a maximal vertex for the objective $c$,
        we must have $c A_{B^0}^{-1} \geq 0$. It follows that the ray $c\R_{\geq 0}$ intersects the facet $F^0$.
        The affine hull of $F^0$ can be described as $\aff(F^0) = \{y \colon y^\T v^+ =1 \}$,
        from which we may observe that $\frac{c}{c^\T v^+} \in \aff(F^0)$.
        Taking the previous two points together we find that $\frac{c}{c^\T v^+} \in F^0$.
        A similar argument gives that $\frac{-c}{(-c)^\T v^-}\in F^k$.
	Thus, we know for the start- and endpoint that
	\begin{align*}
		\frac{c}{c^\T v^+}, a_{p^1} \in F^0, \qquad
		\frac{-c}{(-c)^\T v^-}\in F^k.
	\end{align*}
	For the starting point we conclude
	$\norm{\frac{c}{c^\T v^+} - a_{p^1}} \leq \gamma$.
        For the endpoint we observe that $B^{(d-1)\ell} \cap B^k \neq\emptyset$ and take $p' \in B^\ell \cap B^k$ arbitrarily.
        Making use of the Euclidean diameters of $F^{(d-1)\ell}$ and $F^k$, the triangle inequality gives
        $\norm{\frac{-c}{c^\T v^-} - a_{p^\ell}} \leq \norm{\frac{-c}{c^\T v^-} - a_{p'}} + \norm{a_{p'} - a_{p^\ell}} \leq 2\gamma$.
	Finally note that, by Cauchy-Schwarz, $0 < c^\T v^+ \leq \norm{c} \cdot \norm{v^+} \leq R$ and similarly $0 < (-c)^\T v^- \leq R$.
	We can now use the triangle inequality again to find
	\begin{align*}
            \frac{2}{R} = \frac{\norm{c - (-c)}}{R}
		&\leq \norm{\frac{c}{c^\T v^+} - \frac{-c}{(-c)^\T v^-}} \\
		&\leq \norm{\frac{c}{c^\T v^+} - a_{p^1}} + \norm{a_{p^\ell} - \frac{-c}{(-c)^\T v^+}} + \sum_{t=1}^{\ell-1} \norm{a_{p^t} - a_{p^{t+1}}} \\
		&\leq (\ell+3)\gamma.
	\end{align*}
	Hence we find that $k/(d-1) \geq \lfloor k/(d-1) \rfloor = \ell \geq \frac{2}{R\gamma}-3$
	and $k \geq (d-1)(\frac{2}{R\gamma}-3)$.
    This implies that the sequence of feasible bases must have length at least $(d-1)(\frac{2}{R\gamma}-3)$.
    The sequence was arbitrary, so we find that any simplex path connecting $v^+$ and $v^-$
    is at least this long.
\end{proof}

With these lemmas in place, we can prove our high-probability lower bound
on the diameter of the polyhedron after perturbing.

\begin{theorem}\label{thm:lb-glue}
    Given $d \geq 2$ and $\sigma > 0$ satisfying $\sigma\sqrt{\ln(4/\sigma)} \leq \frac{1}{2600d}$,
    take $n = \lfloor (4/\sigma)^d \rfloor$.
    There exist $\bar A \in \R^{n \times d}$, and $\bar b \in \R^n$ such that the following holds.
    The rows of the combined matrix $(\bar A,\bar b)$ each have norm at most $1$.
    If $ A,  b$ have their entries independently Gaussian distributed
    with variance $\sigma^2$ and expectation $\E[ A] = \bar A, \E[ b] = \bar b$,
	then the combinatorial diameter of the polyhedron $\{x : Ax \leq b\}$
	satisfies
    \[
	    \Pr\left[\diam(\{x : Ax \leq b\}) \geq \frac{(d-1)^{1/2}}{24\sqrt{\sigma\sqrt{\ln(4/\sigma)}}}\right] \geq 1-n^{-d}.
    \]
    Moreover, with probability at least $1-n^{-d}$,
    for any nonzero objective vector $c \in \R^d$,
	its maximizing and minimizing vertices have at least this combinatorial distance.
\end{theorem}
\begin{proof}
    We pick $S \subset \sfe$ to be $\sigma$-dense with $\abs{S} = \lfloor (4/\sigma)^d \rfloor$
    as demonstrated by \Cref{lem:matousek}.
    We set $n = \abs{S}$ and let $\bar A \in \R^{n \times d}$ be formed by having the elements of $S$
    as its rows. We set $\bar b \in \R^n$ to be the all-ones vector,
    sample $ A \in \R^{n \times d}$ and $ b \in \R^n$ with Gaussian distributed
    entries as specified in this theorem's statement and write
    \[
        P = \{x \in \R^d :  Ax \leq  b\}.
    \]
    Using the Gaussian tail bound \Cref{lem:gaussian-tail} we get that, with probability at least $1-n^{-d}$, the rows of $\bar A - A$
    all have norm at most $4 \sigma \sqrt{d \ln n}$ and also $\norm{\bar b - b}_\infty \leq 4\sigma\sqrt{d \ln n}$.
    Note that the fact that $S$ is $\sigma$-dense implies that it is $4\sigma\sqrt{d \ln n}$-dense,
    where $4\sigma\sqrt{d \ln n} \leq 1/8$ by assumption on $\sigma$.
	Abbreviate $\eta = 4\sigma\sqrt{d \ln n}$ and note $\eta = 4d \sigma\sqrt{\ln(4/\sigma)} \leq 1/8$.
    We apply \Cref{lem:lb-helper} to the perturbed data $A,b$ and find
	that the above-mentioned high probability events imply that
\begin{equation}\label{eq:superround}
        (1-2\eta)\ball^d \subseteq P \subseteq (1+4\eta)\ball^d.
\end{equation}

	It remains to show that \eqref{eq:superround} implies that the paths from maximizers to minimizers of any objective are large.
	Note that \eqref{eq:superround} implies, using $0 \leq \eta \leq 1/8$, that
\begin{equation}\label{eq:superroundpolar}
	(1-4\eta)\ball^d \subseteq (1+4\eta)^{-1}\ball^d \subseteq P^\circ \subseteq (1-2\eta)^{-1}\ball^d \subseteq (1+3\eta)\ball^d.
\end{equation}
	We show that all facets of $P^\circ$ have upper bounded geometric diameter.
	Let $F \subset P^\circ$ be an arbitrary facet, and let $y \in F$ denote the minimum-norm point inside the facet.
	Let $v \in F$ be an arbitrary point.
    The optimality condition of $y$ means that $y^\T v \geq \norm{y}^2$, which gives us
    \[
        \norm{y - v}^2 = (y-v)^\T(y-v) = \norm{y}^2 + \norm{v}^2 -2y^\T v \leq \norm{v}^2 - \norm{y}^2.
    \]
	From \eqref{eq:superroundpolar} we know that $\norm{v} \leq (1+3\eta)$ and $\norm{y} \geq (1-4\eta)$, resulting in
    \[
        \norm{v}^2 - \norm{y}^2 \leq (1+6\eta + 9\eta^{2}) - (1 - 8\eta + 16\eta^{2}) \leq 14\eta - 7\eta^{2} \leq 14\eta.
    \]
    We thus found that $\norm{y-v} \leq \sqrt{14\eta}$.
    Since $v \in F$ was arbitrary, we must have for any two points $v,v'\in F$
	that $\norm{v-v'} \leq \norm{v-y} + \norm{y-v'} \leq 2\sqrt{14\eta} \leq 8 \sqrt{\eta}$.
	We have found that the geometric diameter of any facet of $P^\circ$ is at most $8\sqrt{\eta}$.
        We call on \Cref{lem:rel-diam} to find that, assuming \eqref{eq:superround}, for any $c \neq 0$, any path from a maximizer of $c$ to a minimizer of $c$ has combinatorial length at least
        $$(d-1)(\frac{2}{(1+4\eta) \cdot 8\sqrt{\eta}}-3) \geq \sqrt{d-1}\left(\frac{1}{12\sqrt{\sigma \sqrt{\ln(4/\sigma)}}}-3\right).$$
        We finish the argument by observing that
        $\frac{1}{12\sqrt{\sigma\sqrt{\ln(4/\sigma)}}} - 3 \geq \frac{1}{24\sqrt{\sigma\sqrt{\ln(4/\sigma)}}}$.
\end{proof}
We have found that, with probability at least $1-n^{-d}$, for any non-zero objective $c$,
any path from the maximizer of $c$ to the minimizer of $c$
has length at least $\frac{\sqrt{d-1}}{24\sqrt{\sigma\sqrt{\ln(4/\sigma)}}}$.
In particular this is true for the combined semi-random shadow path $P(A,b,-c,Z) \cup P(A,b,Z,c)$.
When this happens, at least one of the paths $P(A,b,-c,Z)$ or $P(A,b,Z,c)$ must have
length at least $\frac{\sqrt{d-1}}{48\sqrt{\sigma\sqrt{\ln(4/\sigma)}}}.$
In particular there must exist a non-zero objective $c$ such that the expectation satisfies
$\E[\abs{P(A,b,c,Z)}] \geq \frac{\sqrt{d-1}}{96\sqrt{\sigma\sqrt{\ln(4/\sigma)}}}$,
which is the statement claimed in the introduction.

This lower bound implies that the upper bound in \Cref{thm:smoothedrhs} has optimal noise dependence up to polylog factors,
in the sense that any upper bound on the shadow path length of the form $\operatorname{poly}(d,\sigma^{-1},\log n)$
must have a monomial term with dependence at least $\sigma^{-1/2}$.

\section{Discussion}
In this paper we have proven that a simplex method equipped with the semi-random shadow vertex pivot rule
has running time of $O(\sigma^{-1/2} d^{11/4} \log(n)^{7/4})$ pivot steps under the smoothed complexity model.
When $n = \lfloor (4/\sigma)^d \rfloor$, we further show that the combinatorial diameter of the feasible set can be as large as
$\Omega(\sigma^{-1/2} d^{1/2} \log(4/\sigma)^{-1/4})$ under the smoothed complexity model.
This establishes the dependence on $\sigma$ tight within a factor of $\log(1/\sigma)^{1/4}$.
The upper and lower bounds do not yet agree on the dependence on $d$ and $\log n$ in the regime of small $\sigma$.

We expect both the upper and lower bound to have room for improvement.
In the regime where $\sigma$ is very large, we would expect the smoothed complexity to approach the average-case complexity.
That is, we expect per \cite{b87} that
\[
    \lim_{\sigma \to \infty} R(n,d,\sigma) = \lim_{\sigma \to \infty} D(n,d,\sigma) = \Theta(d^{3/2} \sqrt{\log n}).
\]
The current upper and lower bounds do not yield this same tight bound, and so can be further improved in the large noise regime.

One place where we expect our upper bound to have room for improvement is \Cref{sub:close-to-fixed}.
The current work shows that (with high probability) no more pivot steps are taken when the intermediate objective $2^k c + Z$ comes within an angle $n^{-2d}\sigma$ of the fixed target objective $c$.
This small angle appears in the running time bound as $\sqrt{k}=O(\sqrt{d\ln n})$.
We expect that a more careful analysis would replace this factor with $O(\sqrt{\ln d})$ or $O(\sqrt{\ln n})$.
One way to achieve this would be to adapt the angle bound from \cite{ST04} to more general inequality constraints.
Ideally, this would allow to bound the \emph{expected} number of pivot steps between fixed objectives with an angle less than $\poly(\sigma/n)$ by a constant.
Moving from a high-probability bound to an expectation bound should sharpen the analysis.

Our lower bounds get weaker when the number of constraints $n$ gets larger.
We expect this to be an artifact of the analysis.
A similar weakening is not present in related works \cite{b87,bdghl21,DGGT16}.

Now that the smoothed analysis upper and lower bounds agree on the exponent of $\sigma$, we believe that it is more fruitful to turn to analysis frameworks beyond smoothed analysis.
By having models that more directly model the algorithm and its inputs, we could further deepen our understanding of the efficiency of the simplex method.
The first steps in this direction are being taken in \cite{bbhk}.

\addcontentsline{toc}{section}{References}
\printbibliography

\end{document}